\documentclass{article}

\PassOptionsToPackage{numbers, sort&compress}{natbib}

\usepackage[final]{neurips_2019}

\usepackage[utf8]{inputenc} 
\usepackage[T1]{fontenc}    
\usepackage{hyperref}       
\usepackage{url}            
\usepackage{booktabs}       
\usepackage{amsfonts}       
\usepackage{nicefrac}       
\usepackage{microtype}      

\usepackage{selectp}

\usepackage{xspace}
\usepackage{appendix}
\usepackage{times}
\usepackage{graphicx}
\usepackage{amsthm}
\usepackage{amsmath,amssymb}
\usepackage{amscd}
\usepackage{parskip}
\usepackage{enumitem}
\usepackage{verbatim}
\usepackage{natbib}
\usepackage{subcaption}
\usepackage{bbm}
\usepackage{tikz-cd}
\usepackage{mathrsfs}
\usepackage[capitalise]{cleveref} 
\crefname{equation}{}{} 
\usepackage{multirow}
\usepackage{tabularx}
\usepackage{setspace}

\usepackage{wrapfig}

\usepackage{chngcntr}
\counterwithin{figure}{section}

\newtheorem{theorem}{Theorem}[section]

\newtheorem{lemma}[theorem]{Lemma}

\usepackage{thm-restate}

\makeatletter
\def\thm@space@setup{%
  \thm@preskip=0.35cm plus 0.1cm minus 0.1cm
  \thm@postskip=\thm@preskip 
}
\makeatother

\usepackage{algorithm}
\usepackage[noend]{algpseudocode}
\usepackage{algorithmicx}
\algnewcommand\AND{\textbf{\ and\ }}
\algnewcommand{\algorithmicor}{\textbf{ or }}
\algnewcommand{\OR}{\algorithmicor}

\usepackage{bm}
\DeclareBoldMathCommand{\balpha}{$\alpha$}
\newcommand{\alpharank}{$\alpha$-Rank\xspace}
\newcommand{\balpharank}{\balpha\textbf{-Rank}\xspace}
\newcommand{\responsegraphucb}{ResponseGraphUCB\xspace}
\newcommand{\rulesep}{\unskip\ \vrule\ }

\newcommand{\ssmone}{\widetilde{\mathbf{C}}}
\newcommand{\ssmtwo}{\widetilde{\mathbf{D}}}

\makeatletter
\newcommand{\multiline}[1]{%
  \begin{tabularx}{\dimexpr\linewidth-\ALG@thistlm}[t]{@{}X@{}}
    #1
  \end{tabularx}
}
\makeatother

\title{Multiagent Evaluation under Incomplete Information}
\author{%
    Mark Rowland$^{1,*}$\\
    {\small \texttt{markrowland@google.com}} \\
    \And 
    Shayegan Omidshafiei$^{2,*}$\\
    {\small \texttt{somidshafiei@google.com}} \\
    \And
    Karl Tuyls$^{2}$\\
    {\small \texttt{karltuyls@google.com}} \\
    \AND
    Julien P{\'e}rolat$^{1}$\\
    {\small \texttt{perolat@google.com}} \\
    \And
    Michal Valko$^{2}$\\
    {\small \texttt{valkom@deepmind.com}} \\
    \And
    Georgios Piliouras$^{3}$\\
    {\small \texttt{georgios@sutd.edu.sg}} \\
    \And
    R{\'e}mi Munos$^{2}$\\
    {\small \texttt{munos@google.com}} \\
    \AND
    {\normalfont $^1$DeepMind London $\hspace{5pt}$ $^2$DeepMind Paris $\hspace{5pt}$ $^3$ Singapore University of Technology and Design} \\
    \AND
    {\normalfont $^*$Equal contributors}
}

\begin{document}

\maketitle

\begin{abstract}
    This paper investigates the evaluation of learned multiagent strategies in the incomplete information setting, which plays a critical role in ranking and training of agents. 
    Traditionally, researchers have relied on Elo ratings for this purpose, with recent works also using methods based on Nash equilibria. 
    Unfortunately, Elo is unable to handle intransitive agent interactions, and other techniques are restricted to zero-sum, two-player settings or are limited by the fact that the Nash equilibrium is intractable to compute. 
    Recently, a ranking method called \alpharank, relying on a new graph-based game-theoretic solution concept, was shown to tractably apply to general games.
    However, evaluations based on Elo or \alpharank typically assume noise-free game outcomes, despite the data often being collected from noisy simulations, making this assumption unrealistic in practice. 
    This paper investigates multiagent evaluation in the incomplete information regime, involving general-sum many-player games with noisy outcomes.
    We derive sample complexity guarantees required to confidently rank agents in this setting. 
    We propose adaptive algorithms for accurate ranking, provide correctness and sample complexity guarantees, then introduce a means of connecting uncertainties in noisy match outcomes to uncertainties in rankings. We evaluate the performance of these approaches in several domains, including Bernoulli games, a soccer meta-game, and Kuhn poker.
\end{abstract}

\section{Introduction}

This paper investigates evaluation of learned multiagent strategies given noisy game outcomes.
The Elo rating system is the predominant approach used to evaluate and rank agents that learn through, e.g., reinforcement learning \cite{elo1978rating,silver2017mastering,silver2018general,mnih:15}. Unfortunately, 
the main caveat with Elo is that it cannot handle intransitive relations between interacting agents, and as such its predictive power is too restrictive to be useful in non-transitive situations (a simple example being the game of \emph{Rock-Paper-Scissors}).
Two recent empirical game-theoretic approaches are \emph{Nash Averaging} \cite{balduzzi2018re} and \emph{\alpharank} \cite{omidshafiei2019alpha}. 
Empirical Game Theory Analysis (EGTA) can be used to evaluate learning agents that interact in large-scale multiagent systems, as it remains largely an open question as to how such agents can be evaluated in a principled manner \citep{Tuyls18,TuylsSym,omidshafiei2019alpha}. 
EGTA has been used to investigate this evaluation problem by deploying empirical or meta-games \citep{Walsh02,Walsh03,Wellman06,Wellman13,PhelpsPM04,PhelpsCMNPS07,TuylsP07}. 
Meta-games abstract away the atomic decisions made in the game and instead focus on interactions of high-level agent strategies, enabling the analysis of large-scale games using game-theoretic techniques. Such games are typically constructed from large amounts of data or simulations.
An \emph{evaluation} of the meta-game then gives a means of comparing the strengths of the various agents interacting in the original game (which might, e.g., form an important part of a training pipeline \citep{jaderberg2018human,jaderberg2017population,silver2017mastering}) or of selecting a final agent after training has taken place (see \cref{fig:noisypayoffs}). 

Both Nash Averaging and \alpharank assume noise-free (i.e., complete) information, and while \alpharank applies to general games, Nash Averaging is restricted to 2-player zero-sum settings. 
Unfortunately, we can seldom expect to observe a noise-free specification of a meta-game in practice, as in large multiagent systems it is unrealistic to expect that the various agents under study will be pitted against all other agents a sufficient number of times to obtain reliable statistics about the meta-payoffs in the empirical game.
While there have been prior inquiries into approximation of equilibria (e.g., Nash) using noisy observations \citep{fearnley2015learning,jordan2008searching}, few have considered evaluation or \emph{ranking} of agents in meta-games with incomplete information \citep{Wellman06,Prakash15}. 
Consider, for instance, a meta-game based on various versions of AlphaGo and prior state-of-the-art agents (e.g., Zen) \cite{Tuyls18,DSilverHMGSDSAPL16}; 
the game outcomes are noisy, and due to computational budget not all agents might play against each other. 
These issues are compounded when the simulations required to construct the empirical meta-game are inherently expensive. 

Motivated by the above issues, this paper contributes to multiagent evaluation under incomplete information. 
As we are interested in general games that go beyond dyadic interactions, we focus on \alpharank.
Our contributions are as follows:
first, we provide sample complexity guarantees describing the number of interactions needed to confidently rank the agents in question; 
second, we introduce adaptive sampling algorithms for selecting agent interactions for the purposes of accurate evaluation; 
third, we develop means of propagating uncertainty in payoffs to uncertainty in agent rankings. 
These contributions enable the principled evaluation of agents in the incomplete information regime.

\begin{figure}[t]
    \centering
    \begin{subfigure}[b]{.3\textwidth}
        \includegraphics[width=\textwidth]{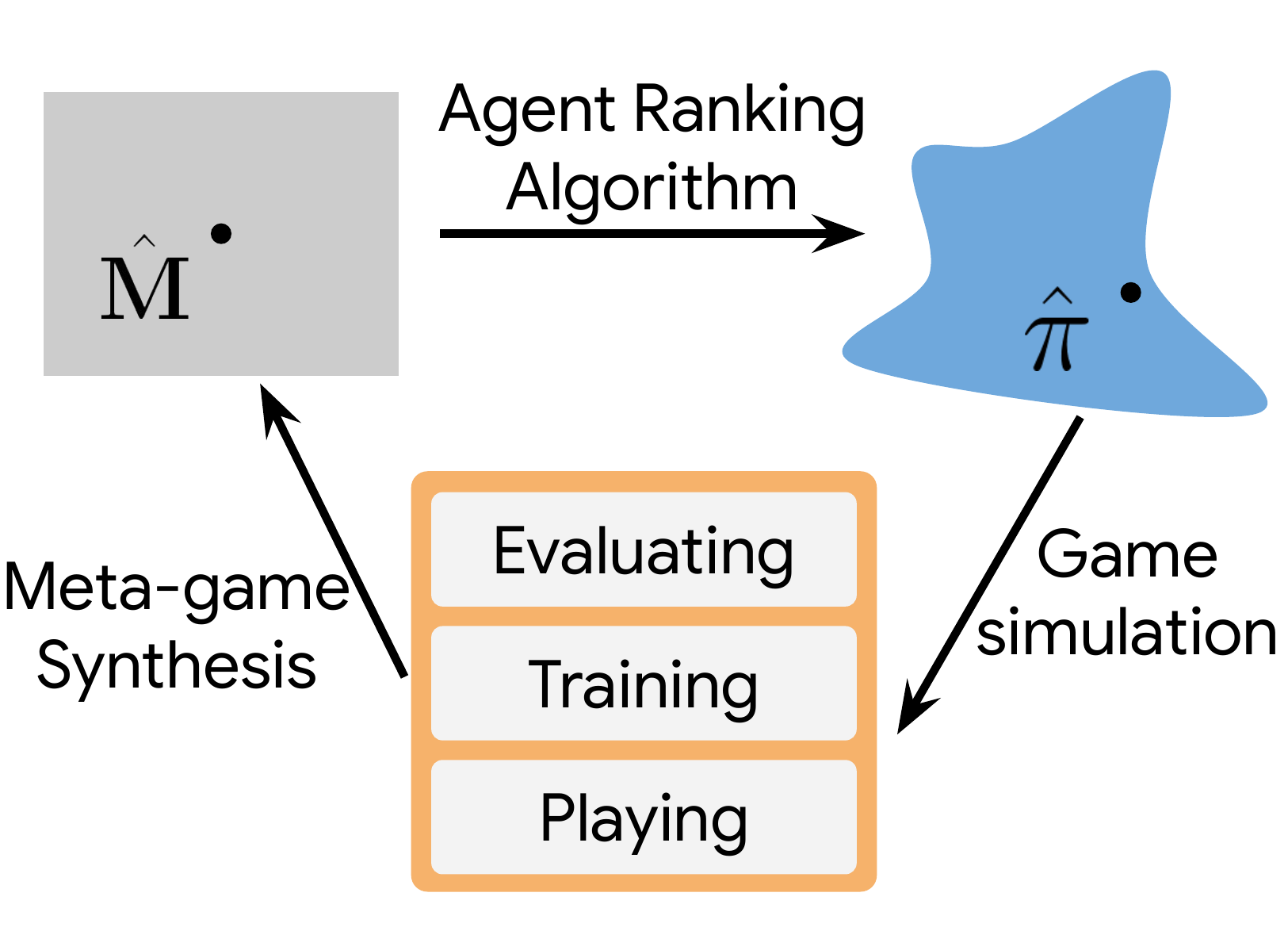}
        \vspace{-16pt}
        \caption{}
        \label{fig:noisypayoffs}
    \end{subfigure}
    \hfill
    \begin{subfigure}[b]{.68\textwidth}
        \includegraphics[width=\textwidth]{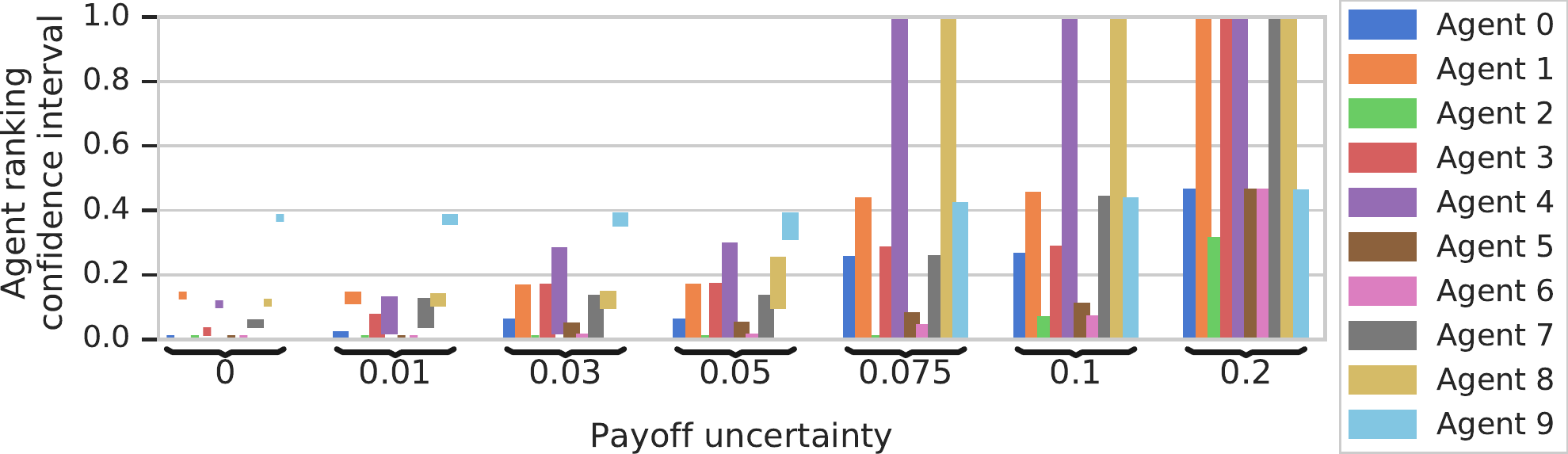}
        \vspace{-16pt}
        \caption{}
        \label{fig:uncertainty_example_soccer}
    \end{subfigure}
    \vspace{-4pt}
    \caption{\subref{fig:noisypayoffs} Illustration of converting plausible payoff matrices consistent with an empirical estimate $\hat{\mathbf{M}}$ to empirical rankings $\hat{\boldsymbol{\pi}}$. The set of plausible payoff matrices and plausible rankings are shown, respectively, in grey and blue. \subref{fig:uncertainty_example_soccer} Ranking uncertainty vs. payoff uncertainty for a soccer meta-game involving 10 agents. Each cluster of bars shows confidence intervals over ranking weights given an observed payoff matrix with a particular uncertainty level; payoff uncertainty here corresponds to the mean confidence interval size of payoff matrix entries.
    This example illustrates the need for careful consideration of payoff uncertainties when computing agent rankings.}
    \label{fig:uncertaintyillustration}
    \vspace{-10pt}
\end{figure}

\section{Preliminaries}\label{sec:preliminaries}
We review here preliminaries in game theory and evaluation.
See Appendix~\ref{sec:related_work} for related work.

\textbf{Games and meta-games. }
Consider a $K$-player game, where each player $k \in [K]$ has a finite set $S^k$ of pure strategies. 
Denote by $S = \prod_k S^k$ the space of pure strategy profiles. 
For each tuple $s = (s^1, \ldots, s^K) \in S$ of pure strategies, the game specifies a joint probability distribution $\nu(s)$ of payoffs to each player. 
The vector of expected payoffs is denoted $\mathbf{M}(s) = (\mathbf{M}^1(s),\ldots,\mathbf{M}^K(s)) \in \mathbb{R}^K$.  
In empirical game theory, we are often interested in analyzing interactions at a higher meta-level, wherein a strategy profile $s$ corresponds to a tuple of machine learning agents and the matrix $\mathbf{M}$ captures their expected payoffs when played against one another in some domain. 
Given this, the notions of `agents' and `strategies' are considered synonymous in this paper. 

\textbf{Evaluation. }
Given payoff matrix $\mathbf{M} \in (\mathbb{R}^K)^{S}$, a key task is to evaluate the strategies in the game. 
This is sometimes done in terms of a game-theoretic solution concept (e.g., Nash equilibria), but may also consist of rankings or numerical scores for strategies.
We focus particularly on the evolutionary dynamics based \alpharank method \citep{omidshafiei2019alpha}, which applies to general many-player games, but also provide supplementary results for the Elo ranking system \citep{elo1978rating}.
There also exist Nash-based evaluation methods, such as Nash Averaging in two-player, zero-sum settings \cite{balduzzi2018re,Tuyls18}, but these are not more generally applicable as the Nash equilibrium is intractable to compute and select \cite{DaskalakisGP06,Harsan88}. 

The exact payoff table $\mathbf{M}$ is rarely known; instead, an empirical payoff table $\hat{\mathbf{M}}$ is typically constructed from observed agent interactions (i.e., samples from the distributions $\nu(s)$). 
Based on collected data, practitioners may associate a set of \emph{plausible} payoff tables with this point estimate, either using a frequentist confidence set, or a Bayesian high posterior density region. 
\Cref{fig:noisypayoffs} illustrates the application of a ranking algorithm to a set of plausible payoff matrices, where rankings can then be used for evaluating, training, or prescribing strategies to play.
\Cref{fig:uncertainty_example_soccer} visualizes an example demonstrating the sensitivity of computed rankings to estimated payoff uncertainties (with ranking uncertainty computed as discussed in \cref{sec:uncertainty}).
This example highlights the importance of propagating payoff uncertainties through to uncertainty in rankings, which can play a critical role, e.g., when allocating training resources to agents based on their respective rankings during learning.

\textbf{\balpharank.} The Elo ranking system (reviewed in Appendix~\ref{sec:elo_rating}) is designed to estimate win-loss probabilities in two-player, symmetric, constant-sum games \citep{elo1978rating}.
Yet despite its widespread use for ranking \citep{lai2015giraffe,arneson2010monte,silver2017mastering,gruslys2017reactor}, Elo has no predictive power in intransitive games (e.g., Rock-Paper-Scissors) \citep{balduzzi2018re}. 
By contrast, \alpharank is a ranking algorithm inspired by evolutionary game theory models, and applies to $K$-player, general-sum games \citep{omidshafiei2019alpha}.
At a high level, \alpharank defines an irreducible Markov chain over strategy set $S$, called the \emph{response graph} of the game \cite{lanctot2017unified}.
The ordered masses of this Markov chain's unique invariant distribution $\boldsymbol{\pi}$ yield the strategy profile rankings.
The Markov transition matrix, $\mathbf{C}$, is defined in a manner that establishes a link to a solution concept called Markov-Conley chains (MCCs).
MCCs are critical for the rankings computed, as they capture agent interactions even under intransitivities and are tractably computed in general games, unlike Nash equilibria \citep{DaskalakisGP06}.

In more detail, the underlying transition matrix over $S$  is defined by \alpharank as follows. 
Let $s=(s^1,\ldots,s^K) \in S$ be a pure strategy profile, and let $\sigma=(\sigma^k, s^{-k})$ be the pure strategy profile which is equal to $s$, except for player $k$, which uses strategy $\sigma^k \in S^k$ instead of $s^k$. 
Denote by $\eta$ the reciprocal of the total number of valid profile transitions from a given strategy profile (i.e., where only a single player deviates in her strategy), so that $\eta = (\sum_{l=1}^K (|S^l| - 1))^{-1}$.
Let $\mathbf{C}_{s, \sigma}$ denote the transition probability from $s$ to $\sigma$, and $\mathbf{C}_{s, s}$ the self-transition probability of $s$, with each defined as:
\begin{align}\label{eq:alpharanktransition1}
    \mkern-28mu \mathbf{C}_{s, \sigma} = 
    \begin{cases}
        \eta \frac{1- \exp\left( -\alpha  \left( \mathbf{M}^k(\sigma) - \mathbf{M}^k(s) \right) \right) }{1- \exp\left( -\alpha m \left( \mathbf{M}^k(\sigma) - \mathbf{M}^k(s) \right) \right)}  & \text{if } \mathbf{M}^k(\sigma) \not= \mathbf{M}^k(s)\\
        \frac{\eta}{m} & \text{otherwise}\, ,
    \end{cases}
    \mkern12mu
    \text{ and }
    \mkern12mu
    \mathbf{C}_{s, s} = 1 - \mkern-32mu \sum_{\substack{k \in [K] \\ \sigma | \sigma^k \in S^k \setminus \{s^k\} } } \mkern-30mu \mathbf{C}_{s , \sigma} \, ,
\end{align}
where if two strategy profiles $s$ and $s^\prime$ differ in more than one player's strategy, then $\mathbf{C}_{s,s^\prime} = 0$. 
Here $\alpha \geq 0$ and $m \in \mathbb{N}$ are parameters to be specified; the form of this transition probability is informed by particular models in evolutionary dynamics and is explained in detail by \citet{omidshafiei2019alpha}, with large values of $\alpha$ corresponding to higher \emph{selection pressure} in the evolutionary model considered.
A key remark is that the correspondence of \alpharank to the MCC solution concept occurs in the limit of infinite $\alpha$. 
In practice, to ensure the irreducibility of $\mathbf{C}$ and the existence of a unique invariant distribution $\boldsymbol{\pi}$, $\alpha$ is either set to a large but finite value, or a perturbed version of $\mathbf{C}$ under the infinite-$\alpha$ limit is used. 
We theoretically and numerically analyze both the finite- and infinite-$\alpha$ regimes in this paper, and provide more details on \alpharank, response graphs, and MCCs in Appendix~\ref{sec:alpharanktheory}.

\section{Sample complexity guarantees}\label{sec:sample_complexity}
This section provides sample complexity bounds, stating the number of strategy profile observations needed to obtain accurate \alpharank rankings with high probability. 
We give two sample complexity results, the first for rankings in the finite-$\alpha$ regime, and the second an instance-dependent guarantee on the reconstruction of the transition matrix in the infinite-$\alpha$ regime.
All proofs are in Appendix~\ref{sec:supp-samplecomplexityproofs}.

\begin{restatable}[Finite-$\alpha$]{theorem}{SampleComplexity}
\label{thm:sample_complexity}
    Suppose payoffs are bounded in the interval $[-M_\mathrm{max}, M_\mathrm{max}]$, and define $L(\alpha, M_{\max}) = 2 \alpha \exp(2\alpha M_\mathrm{max})$ and $g(\alpha, \eta, m, M_\mathrm{max}) = \eta \frac{\exp(2\alpha M_{\mathrm{max}}) - 1}{\exp(2\alpha m M_{\mathrm{max}}) - 1}$.
    Let $\varepsilon \in (0,18\times2^{-|S|}\sum_{n=1}^{|S|-1} \binom{|S|}{n} n^{|S|})$, $\delta \in (0,1)$. 
    Let $\hat{\mathbf{M}}$ be an empirical payoff table constructed by taking $N_s$ i.i.d. interactions of each strategy profile $s \in S$. 
    Then the invariant distribution $\hat{\boldsymbol{\pi}}$ derived from the empirical payoff matrix $\hat{\mathbf{M}}$ satisfies $\max_{s \in \prod_k S^k} |\pi(s) - \hat{\pi}(s) | \leq \varepsilon$ with probability at least $1-\delta$, if
    \begin{align*}
        N_s > \frac{648 M_{\mathrm{max}}^2 \log(2|S|K/\delta) L(\alpha, M_\mathrm{max})^2 \left( \sum_{n=1}^{|S|-1} \binom{|S|}{n} n^{|S|} \right)^2 }{\varepsilon^2 g(\alpha, \eta, m, M_\mathrm{max})^2}  \qquad \forall s \in S \, .
    \end{align*}
\end{restatable}
The dependence on $\delta$ and $\varepsilon$ are as expected from typical Chernoff-style bounds, though Markov chain perturbation theory introduces a dependence on the \alpharank parameters as well, most notably $\alpha$.

\begin{restatable}[Infinite-$\alpha$]{theorem}{ExactRecoverySampleComplexity}\label{thm:sample_complexityinfinite}
    Suppose all payoffs are bounded in $[-M_\mathrm{max}, M_\mathrm{max}]$, and that $\forall k \in [K]$ and $\forall s^{-k} \in S^{-k}$, we have $|\mathbf{M}^k(\sigma, s^{-k}) - \mathbf{M}^k(\tau, s^{-k})| \geq \Delta$ for all distinct $\sigma, \tau \in S^k$, for some $\Delta > 0$. Let $\delta > 0$. Suppose we construct an empirical payoff table $(\hat{\mathbf{M}}^k(s)\ |\ k \in [K], s \in S)$ through $N_s$ i.i.d\,games for each strategy profile $s \in S$. 
    Then the transition matrix $\hat{\mathbf{C}}$ computed from payoff table $\hat{\mathbf{M}}$ is exact (and hence all MCCs are exactly recovered) with probability at least $1-\delta$, if
    \begin{align*}
        N_s >  8\Delta^{-2}M_\mathrm{max}^2 \log(2|S|K/\delta) \qquad \forall s \in S \, .
    \end{align*}
\end{restatable}
A consequence of the theorem is that exact infinite-$\alpha$ rankings are recovered with probability at least $1-\delta$.
We also provide theoretical guarantees for Elo ratings in Appendix~\ref{sec:elo_rating} for completeness.

\section{Adaptive sampling-based ranking}\label{sec:adaptive_sampling_algs}

Whilst instructive, the bounds above have limited utility as the payoff gaps that appear in them are rarely known in practice. 
We next introduce algorithms that compute accurate rankings with high confidence without knowledge of payoff gaps, focusing on \alpharank due to its generality.

\textbf{Problem statement.} Fix an error tolerance $\delta > 0$. 
We seek an algorithm which specifies (i) a sampling scheme $\mathcal{S}$ that selects the next strategy profile $s \in S$ for which a noisy game outcome is observed, and (ii) a criterion $\mathcal{C}(\delta)$ that stops the procedure and outputs the estimated payoff table used for the infinite-$\alpha$ \alpharank rankings, which is exactly correct with probability at least $1-\delta$. 

The assumption of infinite-$\alpha$ simplifies this task; it is sufficient for the algorithm to determine, for each $k \in [K]$ and pair of strategy profiles $(\sigma, s^{-k})$, $(\tau, s^{-k})$, whether $\mathbf{M}^k(\sigma, s^{-k}) > \mathbf{M}^k(\tau, s^{-k})$ or $\mathbf{M}^k(\sigma, s^{-k}) < \mathbf{M}^k(\tau, s^{-k})$ holds. If all such pairwise comparisons are correctly made with probability at least $1-\delta$, the correct rankings can be computed. 
Note that we consider only instances for which the third possibility, $\mathbf{M}^k(\sigma, s^{-k}) = \mathbf{M}^k(\tau, s^{-k})$, does not hold; 
in such cases, it is well-known that it is impossible to design an adaptive strategy that always stops in finite time~\citep{even2006action}.

This problem can be described as a related collection of \emph{pure exploration} bandit problems \citep{bubeck2011pure}; 
each such problem is specified by a  player index $k \in [K]$ and set of two strategy profiles $\{ s, (\sigma^k, s^{-k})\}$ (where $s \in S$, $\sigma^k \in S^k$) that differ only in player $k$; the aim is to determine whether player $k$ receives a greater payoff under strategy profile $s$ or $(\sigma^{k}, s^{-k})$.
Each individual best-arm identification problem can be solved to the required confidence level by maintaining empirical means and a confidence bound for the payoffs concerned. 
Upon termination, an evaluation technique such as \alpharank can then be run on the resulting response graph to compute the strategy profile (or agent) rankings.

\subsection{Algorithm: ResponseGraphUCB}\label{sec:ucb-ue}

We introduce a high-level adaptive sampling algorithm, called \responsegraphucb, for computing accurate rankings in \cref{alg:response-graph-ucb}. 
Several variants of \responsegraphucb are possible, depending on the choice of sampling scheme $\mathcal{S}$ and stopping criterion $\mathcal{C(\delta})$, which we detail next.
\begin{algorithm}[h]
    \caption{\responsegraphucb{}$(\delta,\mathcal{S},\mathcal{C}(\delta))$}
    \label{alg:response-graph-ucb}
    \begin{algorithmic}[1]
            \State{Construct list $L$ of pairs of strategy profiles to compare}
            \State{Initialize tables $\hat{\mathbf{M}}, \mathbf{N}$ to store empirical means and interaction counts}
            \While{$L$ is not empty}
                \State \multiline{Select a strategy profile $s$ appearing in an edge in $L$ using sampling scheme $\mathcal{S}$}
                \State \multiline{Simulate one interaction for $s$ and update $\hat{\mathbf{M}}, \mathbf{N}$ accordingly}
                \State \multiline{Check whether any edges are resolved according to $\mathcal{C}(\delta)$, remove them from $L$ if so}
            \EndWhile 
            \State{\Return empirical table $\hat{\mathbf{M}}$}
    \end{algorithmic}
\end{algorithm}

\textbf{Sampling scheme $\mathcal{S}$. } 
\Cref{alg:response-graph-ucb} keeps track of a list of pairwise strategy profile comparisons that \alpharank requires, removing pairs of profiles for which we have high confidence that the empirical table is correct (according to $\mathcal{C}(\delta)$), and selecting a next strategy profile for simulation.
There are several ways in which strategy profile sampling can be conducted in Algorithm \ref{alg:response-graph-ucb}. 
\textbf{Uniform (U):} A strategy profile is drawn uniformly from all those involved in an unresolved pair. 
\textbf{Uniform-exhaustive (UE):} A pair of strategy profiles is selected uniformly from the set of unresolved pairs, and both strategy profiles are queried until the pair is resolved.
\textbf{Valence-weighted (VW):} As each query of a profile informs multiple payoffs and has impacts on even greater numbers of pairwise comparisons, there may be value in first querying profiles that may resolve a large number of comparisons. Here we set the probability of sampling $s$ proportional to the squared valence of node $s$ in the graph of unresolved comparisons. 
\textbf{Count-weighted (CW):} The marginal impact on the width of a confidence interval for a strategy profile with relatively few queries is greater than for one with many queries, motivating preferential sampling of strategy profiles with low query count.
Here, we preferentially sample the strategy profile with lowest count among all strategy profiles with unresolved comparisons.

\textbf{Stopping condition $\mathcal{C}(\delta)$. } The stopping criteria we consider are based on confidence-bound methods, with the intuition that the algorithm stops only when it has high confidence in all pairwise comparisons made. 
To this end, the algorithm maintains a confidence interval for each of the estimates, and judges a pairwise comparison to be resolved when the two confidence intervals concerned become disjoint. 
There are a variety of confidence bounds that can be maintained, depending on the specifics of the game; we consider \textbf{Hoeffding (UCB)} and \textbf{Clopper-Pearson (CP-UCB)} bounds, along with relaxed variants of each (respectively, \textbf{R-UCB} and \textbf{R-CP-UCB});
full descriptions are given in Appendix~\ref{sec:supp-responsegraphucb}.

\begin{figure}[t]
    \centering
    \begin{subtable}[c]{0.24\textwidth}
        \begin{align*}
            \begin{array}{cc|cc}
            & & \multicolumn{2}{c}{\text{II}} \\
            & & \text{0} & \text{1} \\ \hline
            \multirow{2}{*}{I}  & \text{0} & 0.50,0.50 & 0.85,0.15 \\
             & \text{1} & 0.15,0.85 & 0.50,0.50
            \end{array} \, 
        \end{align*}
        \vspace{-7pt}
        \caption{Players I and II payoffs.}
        \label{table:ucb-ue-payoff-table}
    \end{subtable}
    \begin{subfigure}[c]{.32\textwidth}
        \centering
        \includegraphics[keepaspectratio,width=\textwidth]{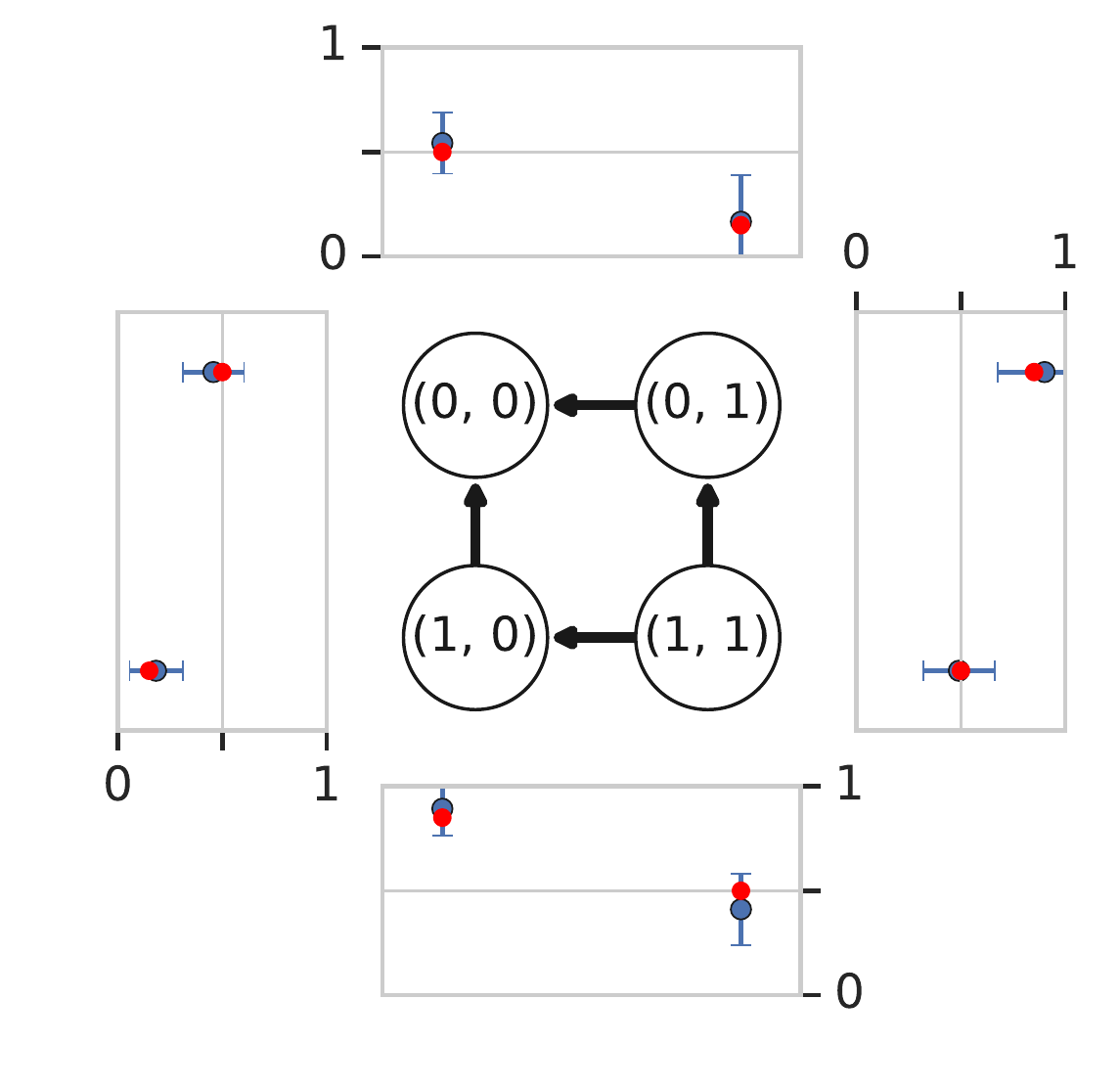}
        \vspace{-17pt}
        \caption{Reconstructed response graph.}
        \label{fig:ucb-ue-result}
    \end{subfigure}
    \hfill%
    \begin{subfigure}[c]{.32\textwidth}
        \centering
        \includegraphics[keepaspectratio,width=\textwidth]{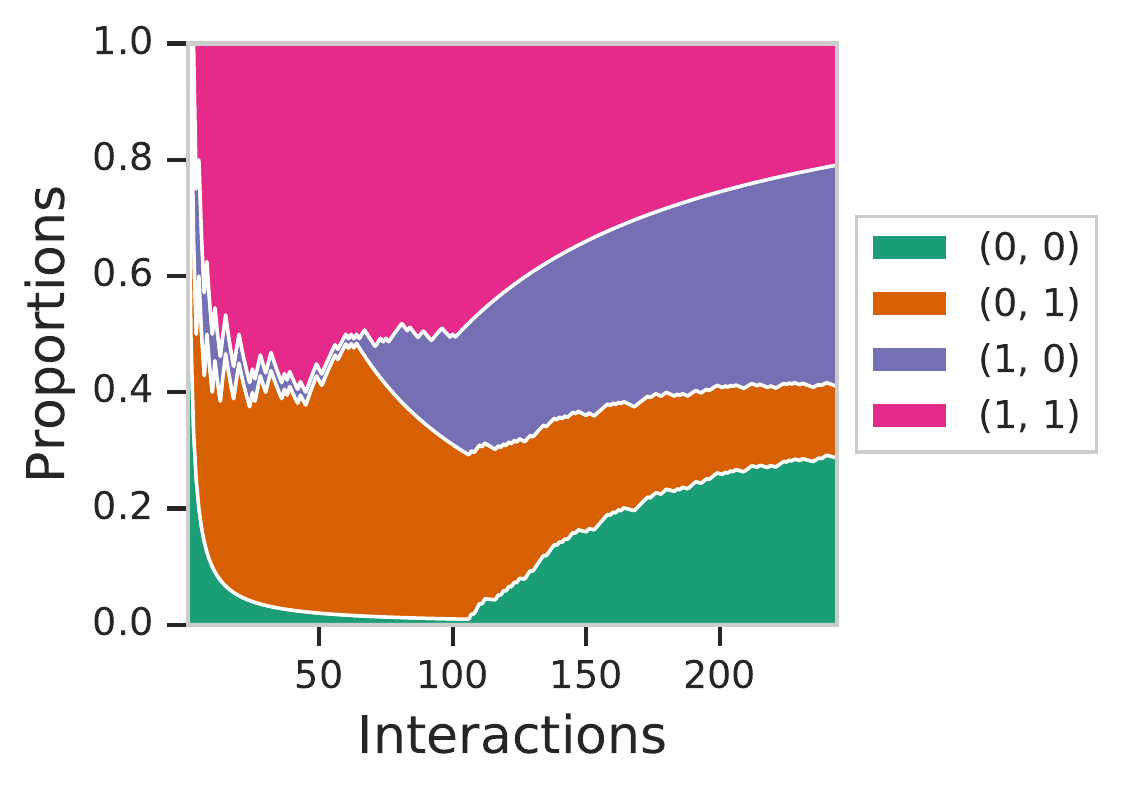}
        \vspace{-15pt}
        \caption{Strategy-wise sample counts.}
        \label{fig:ucb-ue-result-queries}
    \end{subfigure}
    \caption{\responsegraphucb{}($\delta:$ 0.1, $\mathcal{S}$: UE, $\mathcal{C}$: UCB) run on a two-player game. 
    \subref{table:ucb-ue-payoff-table} The payoff tables for both players. 
    \subref{fig:ucb-ue-result} Reconstructed response graph, together with final empirical payoffs and confidence intervals (in blue) and true payoffs (in red).
    \subref{fig:ucb-ue-result-queries} Strategy-wise sample proportions.}
    \label{fig:ucb-ue}
    \vspace{-15pt}
\end{figure}

We build intuition by evaluating \responsegraphucb{}($\delta: 0.1$, $\mathcal{S}:$ UE, $\mathcal{C}:$ UCB), i.e., with a 90\% confidence level, on a two-player game with payoffs shown in \cref{table:ucb-ue-payoff-table}; 
noisy payoffs are simulated as detailed in \cref{sec:experiments}.
The output is given in \cref{fig:ucb-ue-result}; 
the center of this figure shows the estimated response graph, which matches the ground truth in this example.
Around the response graph, mean payoff estimates and confidence bounds are shown for each player-strategy profile combination in blue; 
in each of the surrounding four plots, \responsegraphucb aims to establish which of the true payoffs (shown as red dots) is greater for the deviating player, with directed edges pointing towards estimated higher-payoff deviations. 
\Cref{fig:ucb-ue-result} reveals that strategy profile $(0,0)$ is the sole sink of the response graph, thus would be ranked first by \alpharank.
Each profile has been sampled a different number of times, with running averages of sampling proportions shown in \cref{fig:ucb-ue-result-queries}.
Exploiting knowledge of game symmetry (e.g., as in \cref{table:ucb-ue-payoff-table}) can reduce sample complexity; see Appendix~\ref{sec:symmetric_responsegraphucb}.

We now show the correctness of {\responsegraphucb} and bound the number samples required for it to terminate.
Our analysis depends on the choice of confidence bounds used in stopping condition $\mathcal{C}(\delta)$; we describe the correctness proof in a manner agnostic to these details, and give a sample complexity result for the case of Hoeffding confidence bounds.
See Appendix~\ref{sec:supp-adaptivesampling}
for proofs.

\begin{restatable}{theorem}{thmAdaptiveCorrectness}\label{thm:adaptivecorrectness}
    The {\responsegraphucb} algorithm is correct with high probability: Given $\delta \in (0,1)$, for any particular sampling scheme there is a choice of confidence levels such that ResponseGraphUCB outputs the correct response graph with probability at least $1-\delta$.
\end{restatable}

\begin{restatable}{theorem}{thmAdaptiveSampleComplexity}
    The {\responsegraphucb} algorithm, using confidence parameter $\delta$ and Hoeffding confidence bounds, run on an evaluation instance with $\Delta = \min_{(s^k, s^{-k}), (\sigma^k, s^{-k})} |\mathbf{M}^k(s^k, s^{-k}) - \mathbf{M}^k(\sigma^k, s^{-k})|$ requires at most $\mathcal{O}(\Delta^{-2} \log(1/(\delta\Delta)))$ samples with probability at least $1-2\delta$.
\end{restatable}

\section{Ranking uncertainty propagation}\label{sec:uncertainty}
This section considers the remaining key issue of efficiently computing uncertainty in the ranking weights, given remaining uncertainty in estimated payoffs.
We assume known element-wise upper- and lower-confidence bounds $\mathbf{U}$ and $\mathbf{L}$ on the unknown true payoff table $\mathbf{M}$, e.g., as provided by \responsegraphucb. 
The task we seek to solve is, given a particular strategy profile $s \in S$ and these payoff bounds, to output the confidence interval for $\pi(s)$, the ranking weight for $s$ under the true payoff table $\mathbf{M}$;
i.e., we seek $[\inf_{\mathbf{L} \leq \hat{\mathbf{M}} \leq \mathbf{U}} \pi_{\hat{\mathbf{M}}}(s), \sup_{\mathbf{L} \leq \hat{\mathbf{M}} \leq \mathbf{U}} \pi_{\hat{\mathbf{M}}}(s)]$, where $\pi_{\hat{\mathbf{M}}}$ denotes the output of infinite-$\alpha$ \alpharank under payoffs $\hat{\mathbf{M}}$.
This section proposes an efficient means of solving this task.

At the very highest level, this essentially involves finding plausible response graphs (that is, response graphs that are compatible with a payoff matrix $\hat{\mathbf{M}}$ within the confidence bounds $\mathbf{L}$ and $\mathbf{U}$) that minimize or maximize the probability $\pi(s)$ given to particular strategy profiles $s \in S$ under infinite-$\alpha$ \alpharank. Considering the maximization case, intuitively this may involve directing as many edges adjacent to $s$ towards $s$ as possible, so as to maximize the amount of time the corresponding Markov chain spends at $s$. It is less clear intuitively what the optimal way to set the directions of edges not adjacent to $s$ should be, and how to enforce consistency with the constraints $\mathbf{L} \leq \hat{\mathbf{M}} \leq \mathbf{U}$. In fact, similar problems have been studied before in the PageRank literature for search engine optimization \citep{hollanders2014tight,como2015robustness,csaji2014pagerank,csaji2010pagerank,fercoq2013ergodic}, and have been shown to be reducible to constrained dynamic programming problems.

More formally, the main idea is to convert the problem of obtaining bounds on $\boldsymbol{\pi}$ to a constrained stochastic shortest path (CSSP) policy optimization problem which optimizes \emph{mean return time} for the strategy profile $s$ in the corresponding .
In full generality, such constrained policy optimization problems are known to be NP-hard \citep{csaji2014pagerank}. 
Here, we show that it is sufficient to optimize an \emph{unconstrained} version of the \alpharank CSSP, hence yielding a tractable problem that can be solved with standard SSP optimization routines. Details of the algorithm are provided in Appendix~\ref{sec:uncertaintysupp}; here, we provide a high-level overview of its structure, and state the main theoretical result underlying the correctness of the approach.

The first step is to convert the element-wise confidence bounds $\mathbf{L} \leq \hat{\mathbf{M}} \leq \mathbf{U}$ into a valid set of constraints on the form of the underlying response graph. Next, a reduction is used to encode the problem as policy optimization in a constrained shortest path problem (CSSP), as in the PageRank literature \citep{csaji2014pagerank}; we denote the corresponding problem instance by $\texttt{CSSP}(S, \mathbf{L}, \mathbf{U}, s)$. Whilst solution of CSSPs is in general hard, we note here that it is possible to remove the constraints on the problem, yielding a stochastic shortest path problem that can be solved by standard means.

\begin{restatable}{theorem}{thmCSSP}\label{thm:cssp}
    The unconstrained SSP problem given by removing the action consistency constraints of $\texttt{CSSP}(S, \mathbf{L}, \mathbf{U}, s)$ has the same optimal value as $\texttt{CSSP}(S, \mathbf{L}, \mathbf{U}, s)$.
\end{restatable}
See Appendix~\ref{sec:uncertaintysupp} for the proof. 
Thus, the general approach for finding worst-case upper and lower bounds on infinite-$\alpha$ \alpharank ranking weights $\pi(s)$ for a given strategy profile $s \in S$ is to formulate the unconstrained SSP described above, find the optimal policy (using, e.g., linear programming, policy or value iteration), and then use the inverse relationship between mean return times and stationary distribution probabilities in recurrent Markov chains to obtain the bound on the ranking weight $\pi(s)$ as required; full details are given in Appendix~\ref{sec:uncertaintysupp}.
This approach, when applied to the soccer domain described in the sequel, yields \cref{fig:uncertainty_example_soccer}.

\section{Experiments}\label{sec:experiments}
We consider three domains of increasing complexity, with experimental procedures detailed in Appendix~\ref{sec:experiment_procedures}.
First, we consider randomly-generated two-player zero-sum \textbf{Bernoulli games}, with the constraint that payoffs $\mathbf{M}^k(s,\sigma)\in[0,1]$ cannot be too close to $0.5$ for all pairs of distinct strategies $s,\sigma \in S$ where $\sigma=(\sigma^k, s^{-k})$ (i.e., a single-player deviation from $s$).
This constraint implies that we avoid games that require an exceedingly large number of interactions for the sampler to compute a reasonable estimate of the payoff table.
Second, we analyze a \textbf{Soccer meta-game} with the payoffs in \citet[Figure 2]{liu2018emergent}, wherein agents learn to play soccer in the MuJoCo simulation environment \citep{todorov:12} and are evaluated against one another. 
This corresponds to a two-player symmetric zero-sum game with 10 agents, but with empirical (rather than randomly-generated) payoffs.
Finally, we consider a \textbf{Kuhn poker meta-game} with asymmetric payoffs and 3 players with access to 3 agents each, similar to the domain analyzed in \cite{omidshafiei2019alpha}; 
here, only \alpharank (and not Elo) applies for evaluation due to more than two players being involved.
In all domains, noisy outcomes are simulated by drawing the winning player i.i.d. from a Bernoulli($\mathbf{M}^k(s)$) distribution over payoff tables $\mathbf{M}$. 

\begin{figure}[t]
    \centering
    \null\hfill
    \begin{subfigure}{.33\textwidth}
        \centering
        \includegraphics[width=\textwidth]{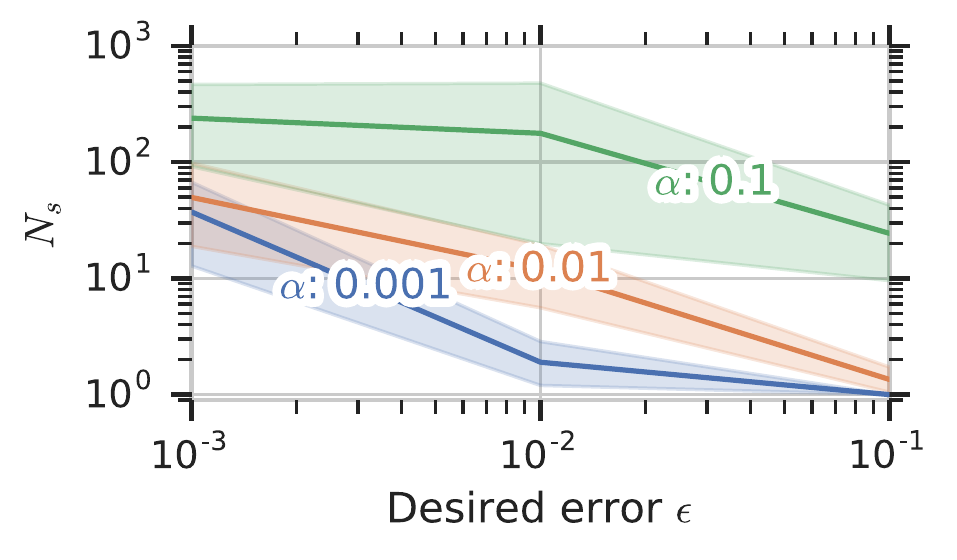}
        \vspace{-17pt}
        \caption{Bernoulli games.}
        \label{fig:alpharank_N_s_empirical_bernoulli}
        \vspace{-3pt}
    \end{subfigure}%
    \hfill
    \begin{subfigure}{.33\textwidth}
        \centering
        \includegraphics[width=\textwidth]{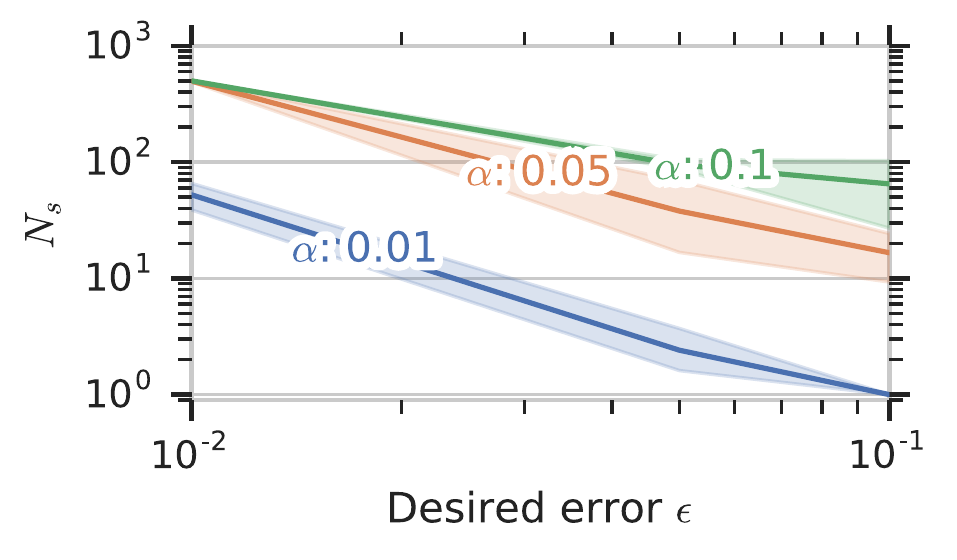}
        \vspace{-17pt}
        \caption{Soccer meta-game.}
        \label{fig:alpharank_N_s_empirical_soccer}
        \vspace{-3pt}
    \end{subfigure}%
    \hfill
    \begin{subfigure}{.33\textwidth}
        \centering
        \includegraphics[width=\textwidth]{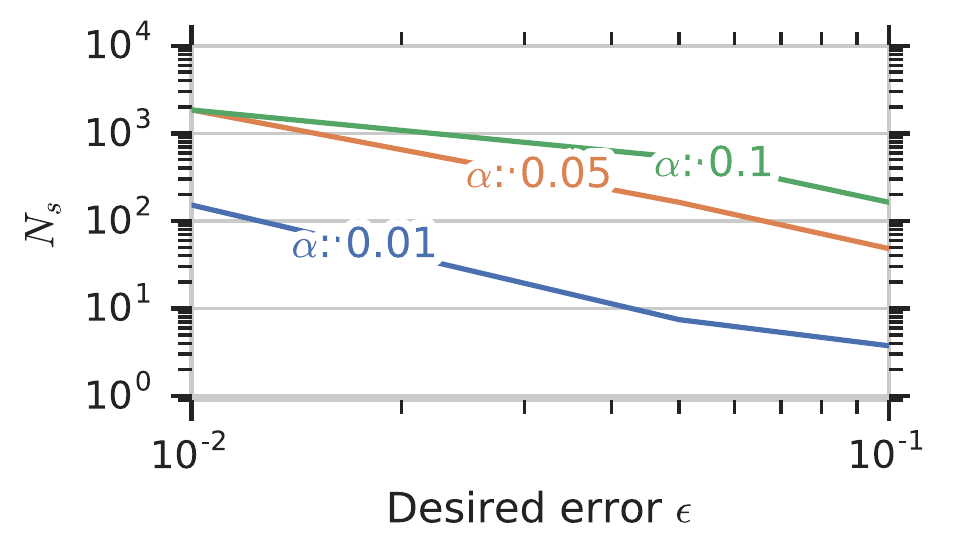}
        \vspace{-17pt}
        \caption{Poker meta-game.}
        \label{fig:alpharank_N_s_empirical_poker}
        \vspace{-3pt}
    \end{subfigure}
    \caption{Samples needed per strategy profile ($N_s$) for finite-$\alpha$ \alpharank, without adaptive sampling.}
    \label{fig:N_s_empirical}
    \vspace{-10pt}
\end{figure}
        
We first consider the empirical sample complexity of \alpharank in the finite-$\alpha$ regime.
\Cref{fig:N_s_empirical} visualizes the number of samples needed per strategy profile to obtain rankings given a desired invariant distribution error $\epsilon$, where $\max_{s \in \prod_k S^k} |\pi(s) - \hat{\pi}(s) | \leq \varepsilon$.
As noted in \cref{thm:sample_complexity}, the sample complexity increases with respect to $\alpha$, with the larger soccer and poker domains requiring on the order of $1\mathrm{e}3$ samples per strategy profile to compute reasonably accurate rankings.  
These results are also intuitive given the evolutionary model underlying \alpharank, where lower $\alpha$ induces lower selection pressure, such that strategies perform almost equally well and are, thus, easier to rank.

As noted in \cref{sec:adaptive_sampling_algs}, sample complexity and ranking error under adaptive sampling are of particular interest. 
To evaluate this, we consider variants of \responsegraphucb in \cref{fig:sweep_sampling_and_confidence}, with particular focus on the UE sampler ($\mathcal{S}$: UE) for visual clarity;
complete results for all combinations of $\mathcal{S}$ and $\mathcal{C}(\delta)$ are presented in Appendix~\cref{sec:full_comparison_plots}.
Consider first the results for the Bernoulli games, shown in \cref{fig:sweep_comparison_bernoulli};
the top row plots the number of interactions required by \responsegraphucb to accurately compute the response graph given a desired error tolerance $\delta$, while the bottom row plots the number of response graph edge errors (i.e., the number of directed edges in the estimated response graph that point in the opposite direction of the ground truth graph). 
Notably, the CP-UCB confidence bound is guaranteed to be tighter than the Hoeffding bounds used in standard UCB, thus the former requires fewer interactions to arrive at a reasonable response graph estimate with the same confidence as the latter;
this is particularly evident for the relaxed variants R-CP-UCB, which require roughly an order of magnitude fewer samples compared to the other sampling schemes, despite achieving a reasonably low response graph error rate.

Consider next the \responsegraphucb results given noisy outcomes for the soccer and poker meta-games, respectively in \cref{fig:sweep_comparison_soccer,fig:sweep_comparison_poker}. 
Due to the much larger strategy spaces of these games, we cap the number of samples available at $1\mathrm{e}5$.
While the results for poker are qualitatively similar to the Bernoulli games, the soccer results are notably different; 
in \cref{fig:sweep_comparison_soccer} (top), the non-relaxed samplers use the entire budget of $1\mathrm{e}5$ interactions, which occurs due to the large strategy space cardinality. 
Specifically, the player-wise strategy size of 10 in the soccer dataset yields a total of 900 two-arm bandit problems to be solved by \responsegraphucb. 
We note also an interesting trend in \cref{fig:sweep_comparison_soccer} (bottom) for the three \responsegraphucb variants ($\mathcal{S}$: UE, $\mathcal{C(\delta)}$: UCB), ($\mathcal{S}$: UE, $\mathcal{C(\delta)}$: R-UCB), and ($\mathcal{S}$: UE, $\mathcal{C(\delta)}$: CP-UCB). 
In the low error tolerance ($\delta$) regime, the uniform-exhaustive strategy used by these three variants implies that \responsegraphucb spends the majority of its sampling budget observing interactions of an extremely small set of strategy profiles, and thus cannot resolve the remaining response graph edges accurately, resulting in high error. 
As error tolerance $\delta$ increases, while the probability of correct resolution of \emph{individual} edges decreases by definition, the earlier stopping time implies that the \responsegraphucb allocates its budget over a larger set of strategies to observe, which subsequently lowers the \emph{total} number of response graph errors.

\begin{figure}[t]
    \newcommand{\figHeight}{8.5\baselineskip}
    \subcaptionbox*{}[\textwidth]{
        \includegraphics[width=0.9\textwidth]{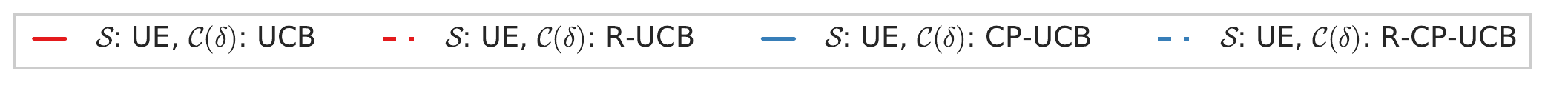}
        \vspace{-19pt}
    }\\
    \subcaptionbox{Bernoulli games.\label{fig:sweep_comparison_bernoulli} 
        }[0.32\textwidth]{ \includegraphics[height=\figHeight]{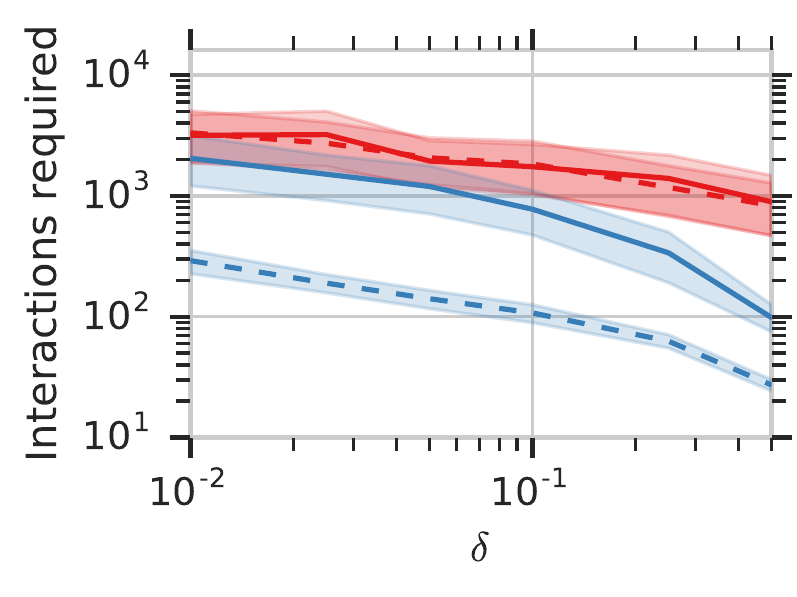}\\[-6pt]
        \includegraphics[height=\figHeight]{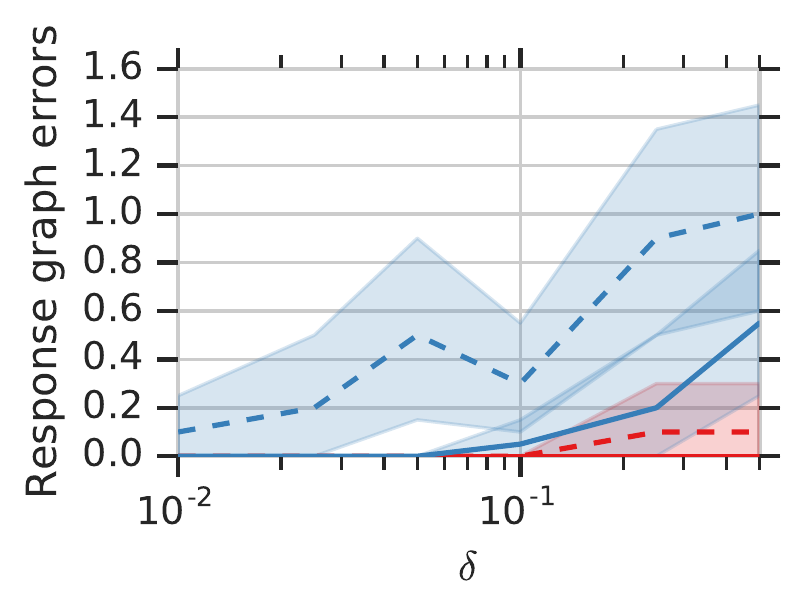}
        \vspace{-8pt}
    }%
    \rulesep%
    \subcaptionbox{Soccer meta-game.\label{fig:sweep_comparison_soccer}
    }[0.32\textwidth]{
        \includegraphics[height=\figHeight]{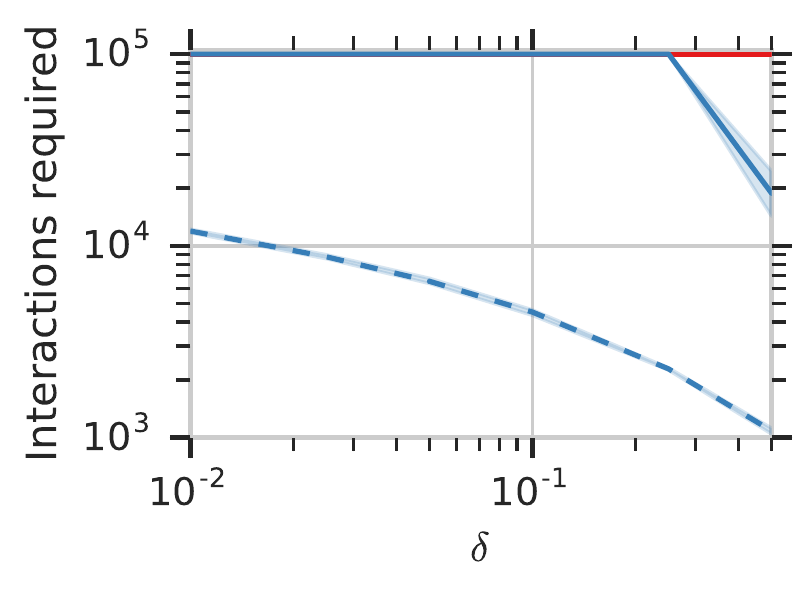}\\[-6pt]
        \includegraphics[height=\figHeight]{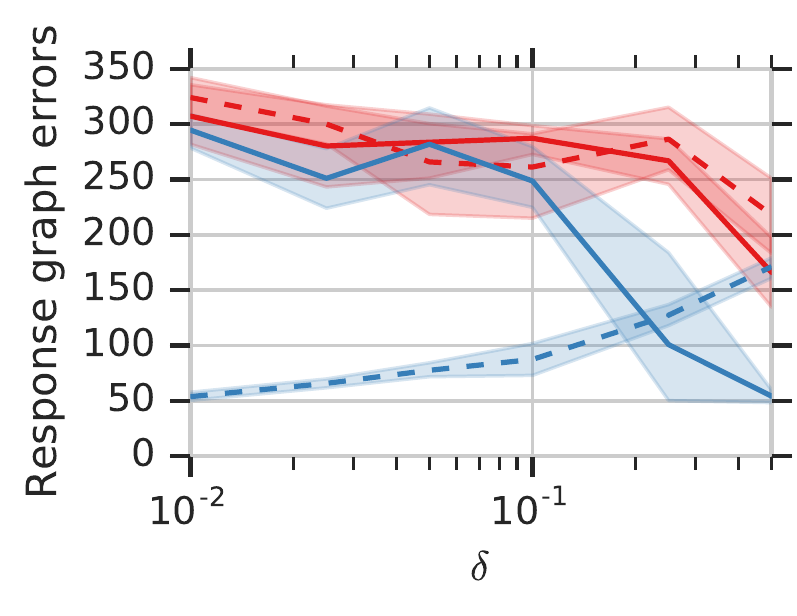}
        \vspace{-8pt}
    }%
    \rulesep%
    \subcaptionbox{Poker meta-game.\label{fig:sweep_comparison_poker}
    }[0.32\textwidth]{
        \includegraphics[height=\figHeight]{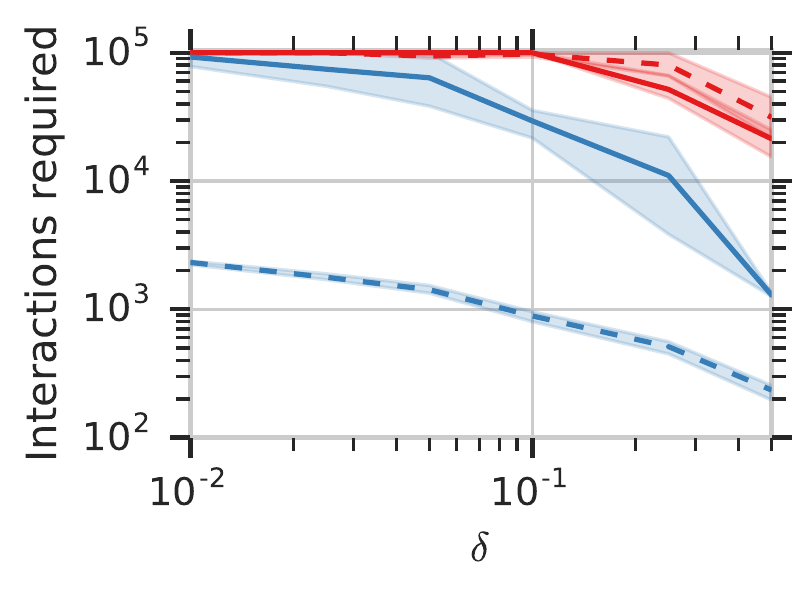}\\[-6pt]
        \includegraphics[height=\figHeight]{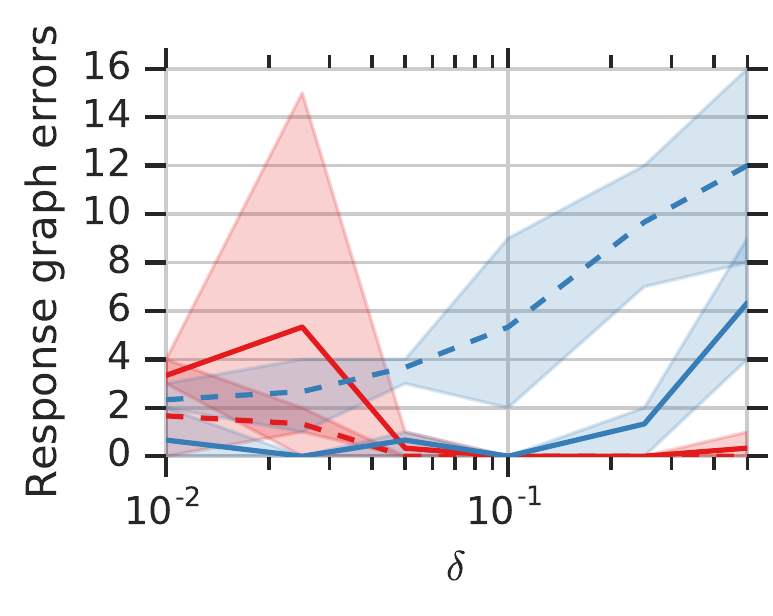}
        \vspace{-8pt}
    }
    \caption{\responsegraphucb performance metrics versus error tolerance $\delta$ for all games. First and second rows, respectively, show the \# of interactions required and response graph edge errors.}
    \label{fig:sweep_sampling_and_confidence}
\end{figure}

\begin{figure}[t]
    \centering
    \includegraphics[width=0.8\linewidth]{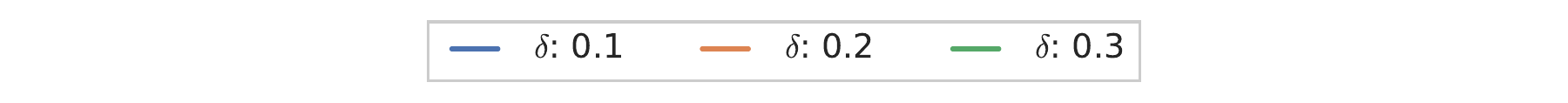}\\[-5pt]
    \subcaptionbox{Soccer meta-game.\label{fig:ranking_error_soccer}}{%
        \vspace{-3pt}%
        \includegraphics[height=7.8\baselineskip]{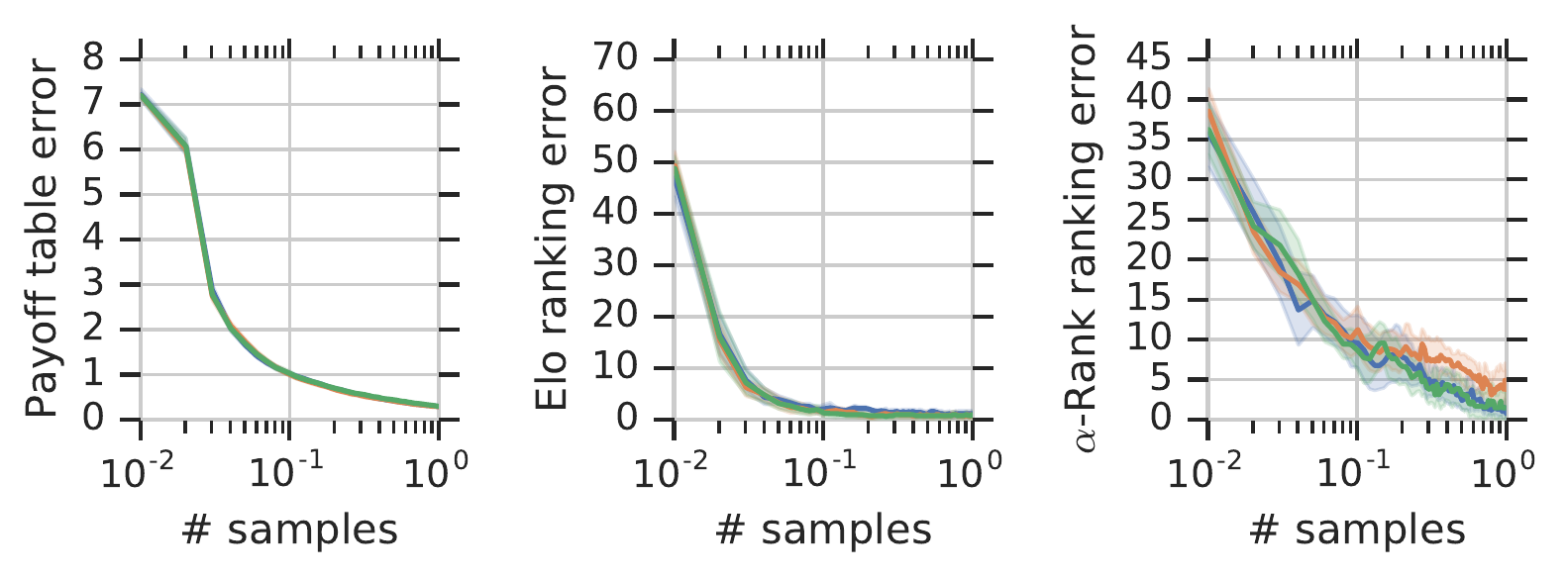}%
        \vspace{-5pt}%
    }%
    \rulesep%
    \subcaptionbox{Poker meta-game.\label{fig:ranking_error_poker}}{%
        \vspace{-3pt}%
        \includegraphics[height=7.8\baselineskip]{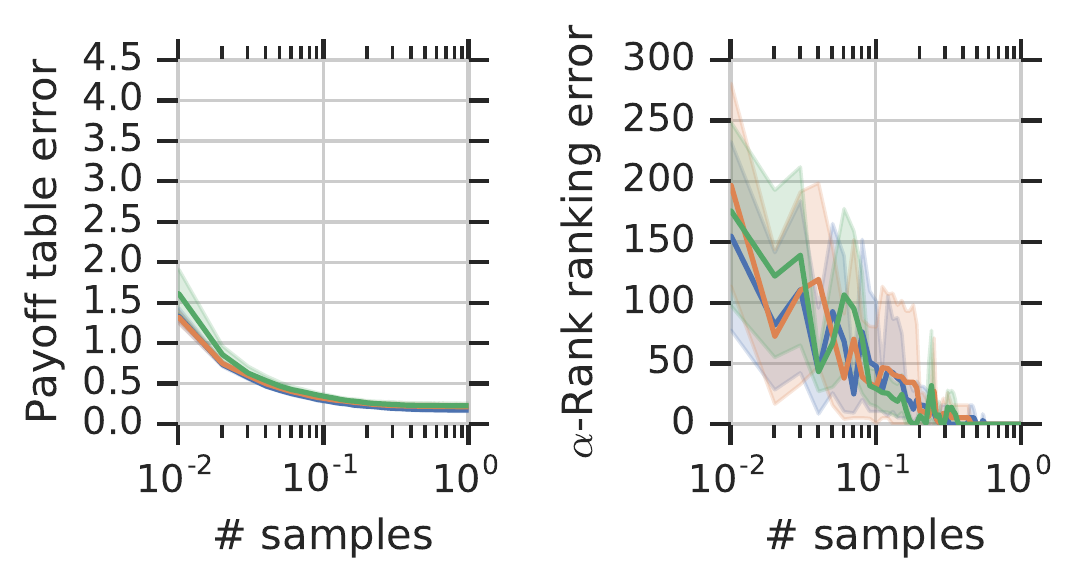}%
        \vspace{-5pt}%
        }
    \vspace{-3pt}
    \caption{Payoff table Frobenius error and ranking errors for various \responsegraphucb confidence levels $\delta$. Number of samples is normalized to $[0,1]$ on the $x$-axis.
    }
    \label{fig:ranking_error_comparisons}
    \vspace{-10pt}
\end{figure}

\begin{figure}
    \centering
    \begin{subfigure}[b]{.5\textwidth}
        \centering
        \includegraphics[width=0.7\textwidth]{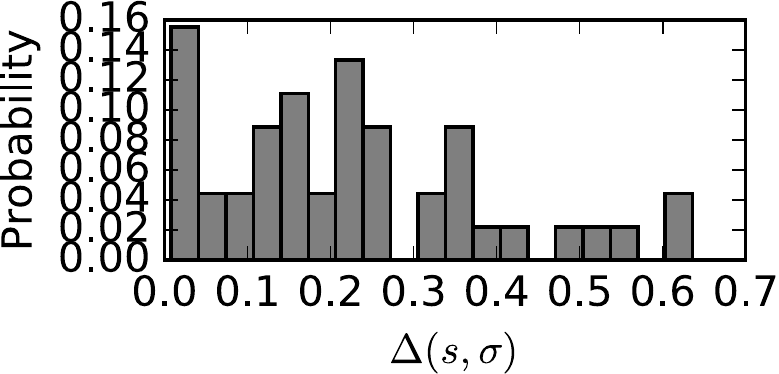}
        \caption{Soccer meta-game.}
        \label{fig:Delta_distribution_soccer}
    \end{subfigure}%
    \hfill%
    \begin{subfigure}[b]{.5\textwidth}
        \centering
        \includegraphics[width=0.7\textwidth]{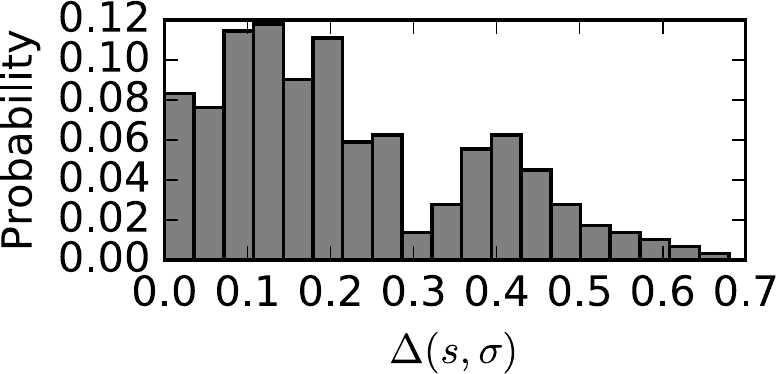}
        \caption{Poker meta-game.}
        \label{img/Delta_distribution_kuhn_poker_3p.pdf}
    \end{subfigure}
    \caption{The ground truth distribution of payoff gaps for all response graph edges in the soccer and poker meta-games.
    We conjecture that the higher ranking variance may be explained by these gaps tending to be more heavily distributed near 0 for poker, making it difficult for \responsegraphucb to sufficiently capture the response graph topology given a high error tolerance $\delta$.}
    \label{fig:Delta_distributions}
\end{figure}

\Cref{fig:ranking_error_soccer} visualizes the ranking errors for Elo and infinite-$\alpha$ \alpharank given various \responsegraphucb error tolerances $\delta$ in the soccer domain.
Ranking errors are computed using the Kendall partial metric (see Appendix~\ref{sec:ranking_metrics}). 
Intuitively, as the estimated payoff table error decreases due to added samples, so does the ranking error for both algorithms.
\Cref{fig:ranking_error_poker} similarly considers the \alpharank ranking error in the poker domain.
Ranking errors again decrease gracefully as the number of samples increases.
Interestingly, while errors are positively correlated with respect to the error tolerances $\delta$ for the poker meta-game, this tolerance parameter seems to have no perceivable effect on the soccer meta-game. 
Moreover, the poker domain results appear to be much higher variance than the soccer counterparts.
To explore this further, we consider the distribution of payoff gaps, which play a key role in determining the response graph reconstruction errors.
Let $\Delta(s, \sigma) = |\mathbf{M}^k(s)-\mathbf{M}^k(\sigma)|$, the payoff difference corresponding to the edge of the response graph where player $k$ deviates, causing a transition between strategy profiles $s,\sigma \in S$.
\Cref{fig:Delta_distributions} plots the ground truth distribution of these gaps for all response graph edges in soccer and poker.
We conjecture that the higher ranking variance may be explained by these gaps tending to be more heavily distributed near 0 for poker, making it difficult for \responsegraphucb to distinguish the `winning' profile and thereby sufficiently capture the response graph topology.

Overall, these results indicate a need for careful consideration of payoff uncertainties when ranking agents, and quantify the effectiveness of the algorithms proposed for multiagent evaluation under incomplete information.
We conclude by remarking that the pairing of bandit algorithms and \alpharank seems a natural means of computing rankings in settings where, e.g., one has a limited budget for adaptively sampling match outcomes.
Our use of bandit algorithms also leads to analysis which is flexible enough to be able to deal with $K$-player general-sum games.
However, approaches such as collaborative filtering may also fare well in their own right.
We conduct a preliminary analysis of this in Appendix~\ref{sec:collab_filtering}, specifically for the case of two-player win-loss games, leaving extensive investigation for follow-up work.

\section{Conclusions}\label{sec:conclusions}
This paper conducted a rigorous investigation of multiagent evaluation under incomplete information.
We focused particularly on \alpharank due to its applicability to general-sum, many-player games.
We provided static sample complexity bounds quantifying the number of interactions needed to confidently rank agents, then introduced several sampling algorithms that adaptively allocate samples to the agent match-ups most informative for ranking. 
We then analyzed the propagation of game outcome uncertainty to the final rankings computed, providing sample complexity guarantees as well as an efficient algorithm for bounding rankings given payoff table uncertainty.
Evaluations were conducted on domains ranging from randomly-generated two-player games to many-player meta-games constructed from real datasets.
The key insight gained by this analysis is that noise in match outcomes plays a prevalent role in determination of agent rankings.
Given the recent emergence of training pipelines  that rely on the evaluation of hundreds of agents pitted against each other in noisy games (e.g., Population-Based Training \citep{jaderberg2017population,jaderberg2018human}), we strongly believe that consideration of these uncertainty sources will play an increasingly important role in multiagent learning.

\section*{Acknowledgements}

We thank Daniel Hennes and Thore Graepel for extensive feedback on an earlier version of this paper, and the anonymous reviewers for their comments and suggestions to improve the paper.
Georgios Piliouras acknowledges MOE AcRF Tier 2 Grant 2016-T2-1-170, grant PIE-SGP-AI-2018-01 and NRF 2018 Fellowship NRF-NRFF2018-07.

\newpage

\bibliographystyle{plainnat}
\bibliography{references}

\begin{thebibliography}{56}
\providecommand{\natexlab}[1]{#1}
\providecommand{\url}[1]{\texttt{#1}}
\expandafter\ifx\csname urlstyle\endcsname\relax
  \providecommand{\doi}[1]{doi: #1}\else
  \providecommand{\doi}{doi: \begingroup \urlstyle{rm}\Url}\fi

\bibitem[Aghassi and Bertsimas(2006)]{Aghassi2006}
Michele Aghassi and Dimitris Bertsimas.
\newblock Robust game theory.
\newblock \emph{Mathematical Programming}, 107\penalty0 (1):\penalty0 231--273,
  Jun 2006.

\bibitem[Arneson et~al.(2010)Arneson, Hayward, and Henderson]{arneson2010monte}
Broderick Arneson, Ryan~B Hayward, and Philip Henderson.
\newblock {M}onte {C}arlo tree search in {H}ex.
\newblock \emph{IEEE Transactions on Computational Intelligence and AI in
  Games}, 2\penalty0 (4):\penalty0 251--258, 2010.

\bibitem[Balduzzi et~al.(2018)Balduzzi, Tuyls, Perolat, and
  Graepel]{balduzzi2018re}
David Balduzzi, Karl Tuyls, Julien Perolat, and Thore Graepel.
\newblock Re-evaluating evaluation.
\newblock In \emph{Advances in Neural Information Processing Systems
  (NeurIPS)}, 2018.

\bibitem[Bubeck et~al.(2011)Bubeck, Munos, and Stoltz]{bubeck2011pure}
S{\'e}bastien Bubeck, R{\'e}mi Munos, and Gilles Stoltz.
\newblock Pure exploration in finitely-armed and continuous-armed bandits.
\newblock \emph{Theoretical Computer Science}, 412\penalty0 (19):\penalty0
  1832--1852, 2011.

\bibitem[Chao et~al.(2018)Chao, Kou, Li, and Peng]{CHAO2018}
Xiangrui Chao, Gang Kou, Tie Li, and Yi~Peng.
\newblock {Jie Ke versus AlphaGo}: A ranking approach using decision making
  method for large-scale data with incomplete information.
\newblock \emph{European Journal of Operational Research}, 265\penalty0
  (1):\penalty0 239--247, 2018.

\bibitem[Clopper and Pearson(1934)]{clopper1934use}
Charles Clopper and Egon Pearson.
\newblock The use of confidence or fiducial limits illustrated in the case of
  the binomial.
\newblock \emph{Biometrika}, 26\penalty0 (4):\penalty0 404--413, 1934.

\bibitem[Como and Fagnani(2015)]{como2015robustness}
Giacomo Como and Fabio Fagnani.
\newblock Robustness of large-scale stochastic matrices to localized
  perturbations.
\newblock \emph{IEEE Transactions on Network Science and Engineering},
  2\penalty0 (2):\penalty0 53--64, 2015.

\bibitem[Coulom(2008)]{Coulom08}
R{\'{e}}mi Coulom.
\newblock Whole-history rating: {A} {B}ayesian rating system for players of
  time-varying strength.
\newblock In \emph{Computers and Games, 6th International Conference, {CG}
  2008, Beijing, China, September 29 - October 1, 2008. Proceedings}, pages
  113--124, 2008.

\bibitem[Cs{\'a}ji et~al.(2010)Cs{\'a}ji, Jungers, and
  Blondel]{csaji2010pagerank}
Bal{\'a}zs~Csan{\'a}d Cs{\'a}ji, Rapha{\"e}l~M Jungers, and Vincent~D Blondel.
\newblock {PageRank} optimization in polynomial time by stochastic shortest
  path reformulation.
\newblock In \emph{International Conference on Algorithmic Learning Theory}.
  Springer, 2010.

\bibitem[Cs{\'a}ji et~al.(2014)Cs{\'a}ji, Jungers, and
  Blondel]{csaji2014pagerank}
Bal{\'a}zs~Csan{\'a}d Cs{\'a}ji, Rapha{\"e}l~M Jungers, and Vincent~D Blondel.
\newblock {PageRank} optimization by edge selection.
\newblock \emph{Discrete Applied Mathematics}, 169:\penalty0 73--87, 2014.

\bibitem[Daskalakis et~al.(2006)Daskalakis, Goldberg, and
  Papadimitriou]{DaskalakisGP06}
Constantinos Daskalakis, Paul~W. Goldberg, and Christos~H. Papadimitriou.
\newblock The complexity of computing a {Nash} equilibrium.
\newblock In \emph{Proceedings of the 38th Annual {ACM} Symposium on Theory of
  Computing, Seattle, WA, USA, May 21-23, 2006}, pages 71--78, 2006.

\bibitem[Elo(1978)]{elo1978rating}
Arpad~E Elo.
\newblock \emph{The Rating of Chessplayers, Past and Present}.
\newblock Arco Pub., 1978.

\bibitem[Even-Dar et~al.(2006)Even-Dar, Mannor, and Mansour]{even2006action}
Eyal Even-Dar, Shie Mannor, and Yishay Mansour.
\newblock Action elimination and stopping conditions for the multi-armed bandit
  and reinforcement learning problems.
\newblock \emph{Journal of Machine Learning Research}, 7\penalty0
  (Jun):\penalty0 1079--1105, 2006.

\bibitem[Fagin et~al.(2006)Fagin, Kumar, Mahdian, Sivakumar, and
  Vee]{fagin2006comparing}
Ronald Fagin, Ravi Kumar, Mohammad Mahdian, D~Sivakumar, and Erik Vee.
\newblock Comparing partial rankings.
\newblock \emph{SIAM Journal on Discrete Mathematics}, 20\penalty0
  (3):\penalty0 628--648, 2006.

\bibitem[Fearnley et~al.(2015)Fearnley, Gairing, Goldberg, and
  Savani]{fearnley2015learning}
John Fearnley, Martin Gairing, Paul~W Goldberg, and Rahul Savani.
\newblock Learning equilibria of games via payoff queries.
\newblock \emph{Journal of Machine Learning Research}, 16\penalty0
  (1):\penalty0 1305--1344, 2015.

\bibitem[Fercoq et~al.(2013)Fercoq, Akian, Bouhtou, and
  Gaubert]{fercoq2013ergodic}
Olivier Fercoq, Marianne Akian, Mustapha Bouhtou, and St{\'e}phane Gaubert.
\newblock Ergodic control and polyhedral approaches to {PageRank} optimization.
\newblock \emph{IEEE Transactions on Automatic Control}, 58\penalty0
  (1):\penalty0 134--148, 2013.

\bibitem[Gabillon et~al.(2012)Gabillon, Ghavamzadeh, and
  Lazaric]{gabillon2012best}
Victor Gabillon, Mohammad Ghavamzadeh, and Alessandro Lazaric.
\newblock Best arm identification: A unified approach to fixed budget and fixed
  confidence.
\newblock In \emph{Advances in Neural Information Processing Systems
  (NeurIPS)}. 2012.

\bibitem[Garivier and Capp{\'e}(2011)]{garivier2011kl}
Aur{\'e}lien Garivier and Olivier Capp{\'e}.
\newblock The {KL-UCB} algorithm for bounded stochastic bandits and beyond.
\newblock In \emph{Proceedings of the Conference on Learning Theory (COLT)},
  2011.

\bibitem[Gruslys et~al.(2018)Gruslys, Dabney, Azar, Piot, Bellemare, and
  Munos]{gruslys2017reactor}
Audrunas Gruslys, Will Dabney, Mohammad~Gheshlaghi Azar, Bilal Piot, Marc
  Bellemare, and R{\'e}mi Munos.
\newblock The reactor: A sample-efficient actor-critic architecture.
\newblock \emph{Proceedings of the International Conference on Learning
  Representations (ICLR)}, 2018.

\bibitem[Harsanyi and Selten(1988)]{Harsan88}
John Harsanyi and Reinhard Selten.
\newblock \emph{A General Theory of Equilibrium Selection in Games}, volume~1.
\newblock The MIT Press, 1 edition, 1988.

\bibitem[Heinrich et~al.(2015)Heinrich, Lanctot, and Silver]{Heinrich15FSP}
Johannes Heinrich, Marc Lanctot, and David Silver.
\newblock Fictitious self-play in extensive-form games.
\newblock In \emph{International Conference on Machine Learning}, 2015.

\bibitem[Hennes et~al.(2013)Hennes, Claes, and Tuyls]{Hennes13}
Daniel Hennes, Daniel Claes, and Karl Tuyls.
\newblock Evolutionary advantage of reciprocity in collision avoidance.
\newblock In \emph{AAMAS Workshop on Autonomous Robots and Multirobot Systems},
  2013.

\bibitem[Herbrich et~al.(2007)Herbrich, Minka, and Graepel]{herbrich:07}
Ralf Herbrich, Tom Minka, and Thore Graepel.
\newblock {TrueSkill: A {B}ayesian skill rating system}.
\newblock In \emph{Advances in Neural Information Processing Systems (NIPS)},
  2007.

\bibitem[Hollanders et~al.(2014)Hollanders, Como, Delvenne, and
  Jungers]{hollanders2014tight}
Romain Hollanders, Giacomo Como, Jean-Charles Delvenne, and Rapha{\"e}l~M
  Jungers.
\newblock Tight bounds on sparse perturbations of {M}arkov chains.
\newblock In \emph{International Symposium on Mathematical Theory of Networks
  and Systems}, 2014.

\bibitem[Jaderberg et~al.(2017)Jaderberg, Dalibard, Osindero, Czarnecki,
  Donahue, Razavi, Vinyals, Green, Dunning, Simonyan,
  et~al.]{jaderberg2017population}
Max Jaderberg, Valentin Dalibard, Simon Osindero, Wojciech~M Czarnecki, Jeff
  Donahue, Ali Razavi, Oriol Vinyals, Tim Green, Iain Dunning, Karen Simonyan,
  et~al.
\newblock Population based training of neural networks.
\newblock \emph{arXiv preprint arXiv:1711.09846}, 2017.

\bibitem[Jaderberg et~al.(2019)Jaderberg, Czarnecki, Dunning, Marris, Lever,
  Casta{\~n}eda, Beattie, Rabinowitz, Morcos, Ruderman, Sonnerat, Green,
  Deason, Leibo, Silver, Hassabis, Kavukcuoglu, and
  Graepel]{jaderberg2018human}
Max Jaderberg, Wojciech~M. Czarnecki, Iain Dunning, Luke Marris, Guy Lever,
  Antonio~Garcia Casta{\~n}eda, Charles Beattie, Neil~C. Rabinowitz, Ari~S.
  Morcos, Avraham Ruderman, Nicolas Sonnerat, Tim Green, Louise Deason, Joel~Z.
  Leibo, David Silver, Demis Hassabis, Koray Kavukcuoglu, and Thore Graepel.
\newblock Human-level performance in 3d multiplayer games with population-based
  reinforcement learning.
\newblock \emph{Science}, 364\penalty0 (6443):\penalty0 859--865, 2019.

\bibitem[Jain et~al.(2013)Jain, Netrapalli, and Sanghavi]{jain2013low}
Prateek Jain, Praneeth Netrapalli, and Sujay Sanghavi.
\newblock Low-rank matrix completion using alternating minimization.
\newblock In \emph{Proceedings of the ACM Symposium on Theory of Computing
  (STOC)}, 2013.

\bibitem[Jordan et~al.(2008)Jordan, Vorobeychik, and
  Wellman]{jordan2008searching}
Patrick~R Jordan, Yevgeniy Vorobeychik, and Michael~P Wellman.
\newblock Searching for approximate equilibria in empirical games.
\newblock In \emph{Proceedings of the International Conference on Autonomous
  Agents and Multiagent Systems (AAMAS)}, 2008.

\bibitem[Kalyanakrishnan et~al.(2012)Kalyanakrishnan, Tewari, Auer, and
  Stone]{kalyanakrishnan2012pac}
Shivaram Kalyanakrishnan, Ambuj Tewari, Peter Auer, and Peter Stone.
\newblock {PAC} subset selection in stochastic multi-armed bandits.
\newblock In \emph{Proceedings of the International Conference on Machine
  Learning (ICML)}, 2012.

\bibitem[Karnin et~al.(2013)Karnin, Koren, and Somekh]{karnin2013almost}
Zohar Karnin, Tomer Koren, and Oren Somekh.
\newblock Almost optimal exploration in multi-armed bandits.
\newblock In \emph{Proceedings of the International Conference on Machine
  Learning (ICML)}, 2013.

\bibitem[Lai(2015)]{lai2015giraffe}
Matthew Lai.
\newblock Giraffe: Using deep reinforcement learning to play chess.
\newblock \emph{arXiv preprint arXiv:1509.01549}, 2015.

\bibitem[Lanctot et~al.(2017)Lanctot, Zambaldi, Gruslys, Lazaridou, Tuyls,
  P{\'e}rolat, Silver, and Graepel]{lanctot2017unified}
Marc Lanctot, Vinicius Zambaldi, Audrunas Gruslys, Angeliki Lazaridou, Karl
  Tuyls, Julien P{\'e}rolat, David Silver, and Thore Graepel.
\newblock A unified game-theoretic approach to multiagent reinforcement
  learning.
\newblock In \emph{Advances in Neural Information Processing Systems (NIPS)},
  2017.

\bibitem[Liu et~al.(2019)Liu, Lever, Heess, Merel, Tunyasuvunakool, and
  Graepel]{liu2018emergent}
Siqi Liu, Guy Lever, Nicholas Heess, Josh Merel, Saran Tunyasuvunakool, and
  Thore Graepel.
\newblock Emergent coordination through competition.
\newblock In \emph{Proceedings of the International Conference on Learning
  Representations (ICLR)}, 2019.

\bibitem[McMahan et~al.(2003)McMahan, Gordon, and Blum]{McMahanGB03}
H.~Brendan McMahan, Geoffrey~J. Gordon, and Avrim Blum.
\newblock Planning in the presence of cost functions controlled by an
  adversary.
\newblock In \emph{Proceedings of the International Conference on Machine
  Learning (ICML)}, 2003.

\bibitem[Mnih et~al.(2015)Mnih, Kavukcuoglu, Silver, Rusu, Veness, Bellemare,
  Graves, Riedmiller, Fidjeland, Ostrovski, Petersen, Beattie, Sadik,
  Antonoglou, King, Kumaran, Wierstra, Legg, and Hassabis]{mnih:15}
Volodymyr Mnih, Koray Kavukcuoglu, David Silver, Andrei~A. Rusu, Joel Veness,
  Marc~G. Bellemare, Alex Graves, Martin Riedmiller, Andreas~K. Fidjeland,
  Georg Ostrovski, Stig Petersen, Charles Beattie, Amir Sadik, Ioannis
  Antonoglou, Helen King, Dharshan Kumaran, Daan Wierstra, Shane Legg, and
  Demis Hassabis.
\newblock Human-level control through deep reinforcement learning.
\newblock \emph{Nature}, 518\penalty0 (7540):\penalty0 529--533, 02 2015.

\bibitem[Omidshafiei et~al.(2019)Omidshafiei, Papadimitriou, Piliouras, Tuyls,
  Rowland, Lespiau, Czarnecki, Lanctot, Perolat, and
  Munos]{omidshafiei2019alpha}
Shayegan Omidshafiei, Christos Papadimitriou, Georgios Piliouras, Karl Tuyls,
  Mark Rowland, Jean-Baptiste Lespiau, Wojciech~M Czarnecki, Marc Lanctot,
  Julien Perolat, and R{\'e}mi Munos.
\newblock {$\alpha$}-{R}ank: Multi-agent evaluation by evolution.
\newblock \emph{Scientific Reports}, 9, 2019.

\bibitem[Phelps et~al.(2004)Phelps, Parsons, and McBurney]{PhelpsPM04}
Steve Phelps, Simon Parsons, and Peter McBurney.
\newblock An evolutionary game-theoretic comparison of two double-auction
  market designs.
\newblock In \emph{Agent-Mediated Electronic Commerce VI, Theories for and
  Engineering of Distributed Mechanisms and Systems, {AAMAS} 2004 Workshop},
  2004.

\bibitem[Phelps et~al.(2007)Phelps, Cai, McBurney, Niu, Parsons, and
  Sklar]{PhelpsCMNPS07}
Steve Phelps, Kai Cai, Peter McBurney, Jinzhong Niu, Simon Parsons, and
  Elizabeth Sklar.
\newblock Auctions, evolution, and multi-agent learning.
\newblock In \emph{Adaptive Agents and Multi-Agent Systems {III.} Adaptation
  and Multi-Agent Learning, 5th, 6th, and 7th European Symposium, {ALAMAS}
  2005-2007 on Adaptive and Learning Agents and Multi-Agent Systems}, 2007.

\bibitem[Ponsen et~al.(2009)Ponsen, Tuyls, Kaisers, and Ramon]{PonsenTKR09}
Marc J.~V. Ponsen, Karl Tuyls, Michael Kaisers, and Jan Ramon.
\newblock An evolutionary game-theoretic analysis of poker strategies.
\newblock \emph{Entertainment Computing}, 1\penalty0 (1):\penalty0 39--45,
  2009.

\bibitem[Prakash and Wellman(2015)]{Prakash15}
Achintya Prakash and Michael~P. Wellman.
\newblock Empirical game-theoretic analysis for moving target defense.
\newblock In \emph{Proceedings of the Second ACM Workshop on Moving Target
  Defense}, 2015.

\bibitem[Silver et~al.(2016)Silver, Huang, Maddison, Guez, Sifre, van~den
  Driessche, Schrittwieser, Antonoglou, Panneershelvam, Lanctot, Dieleman,
  Grewe, Nham, Kalchbrenner, Sutskever, Lillicrap, Leach, Kavukcuoglu, Graepel,
  and Hassabis]{DSilverHMGSDSAPL16}
David Silver, Aja Huang, Chris~J. Maddison, Arthur Guez, Laurent Sifre, George
  van~den Driessche, Julian Schrittwieser, Ioannis Antonoglou, Vedavyas
  Panneershelvam, Marc Lanctot, Sander Dieleman, Dominik Grewe, John Nham, Nal
  Kalchbrenner, Ilya Sutskever, Timothy~P. Lillicrap, Madeleine Leach, Koray
  Kavukcuoglu, Thore Graepel, and Demis Hassabis.
\newblock {Mastering the game of Go with deep neural networks and tree search}.
\newblock \emph{Nature}, 529\penalty0 (7587):\penalty0 484--489, 2016.

\bibitem[Silver et~al.(2017)Silver, Schrittwieser, Simonyan, Antonoglou, Huang,
  Guez, Hubert, Baker, Lai, Bolton, et~al.]{silver2017mastering}
David Silver, Julian Schrittwieser, Karen Simonyan, Ioannis Antonoglou, Aja
  Huang, Arthur Guez, Thomas Hubert, Lucas Baker, Matthew Lai, Adrian Bolton,
  et~al.
\newblock Mastering the game of go without human knowledge.
\newblock \emph{Nature}, 550\penalty0 (7676):\penalty0 354, 2017.

\bibitem[Silver et~al.(2018)Silver, Hubert, Schrittwieser, Antonoglou, Lai,
  Guez, Lanctot, Sifre, Kumaran, Graepel, Lillicrap, Simonyan, and
  Hassabis]{silver2018general}
David Silver, Thomas Hubert, Julian Schrittwieser, Ioannis Antonoglou, Matthew
  Lai, Arthur Guez, Marc Lanctot, Laurent Sifre, Dharshan Kumaran, Thore
  Graepel, Timothy Lillicrap, Karen Simonyan, and Demis Hassabis.
\newblock A general reinforcement learning algorithm that masters chess, shogi,
  and go through self-play.
\newblock \emph{Science}, 362\penalty0 (6419):\penalty0 1140--1144, 2018.

\bibitem[Sokota et~al.(2019)Sokota, Ho, and Wiedenbeck]{sakota2019learning}
Samuel Sokota, Caleb Ho, and Bryce Wiedenbeck.
\newblock Learning deviation payoffs in simulation based games.
\newblock In \emph{AAAI Conference on Artificial Intelligence}, 2019.

\bibitem[Solan and Vieille(2003)]{PerturbedMarkovChains}
Eilon Solan and Nicolas Vieille.
\newblock Perturbed {Markov} chains.
\newblock \emph{Journal of Applied Probability}, 40\penalty0 (1):\penalty0
  107--122, 2003.

\bibitem[Todorov et~al.(2012)Todorov, Erez, and Tassa]{todorov:12}
Emanuel Todorov, Tom Erez, and Yuval Tassa.
\newblock {MuJoCo: A physics engine for model-based control}.
\newblock In \emph{Proceedings of the International Conference on Intelligent
  Robots and Systems (IROS)}, 2012.

\bibitem[Tuyls and Parsons(2007)]{TuylsP07}
Karl Tuyls and Simon Parsons.
\newblock What evolutionary game theory tells us about multiagent learning.
\newblock \emph{Artif. Intell.}, 171\penalty0 (7):\penalty0 406--416, 2007.

\bibitem[Tuyls et~al.(2018{\natexlab{a}})Tuyls, Perolat, Lanctot, Leibo, and
  Graepel]{Tuyls18}
Karl Tuyls, Julien Perolat, Marc Lanctot, Joel~Z Leibo, and Thore Graepel.
\newblock A generalised method for empirical game theoretic analysis.
\newblock In \emph{Proceedings of the International Conference on Autonomous
  Agents and Multiagent Systems (AAMAS)}, 2018{\natexlab{a}}.

\bibitem[Tuyls et~al.(2018{\natexlab{b}})Tuyls, Perolat, Lanctot, Savani,
  Leibo, Ord, Graepel, and Legg]{TuylsSym}
Karl Tuyls, Julien Perolat, Marc Lanctot, Rahul Savani, Joel Leibo, Toby Ord,
  Thore Graepel, and Shane Legg.
\newblock Symmetric decomposition of asymmetric games.
\newblock \emph{Scientific Reports}, 8\penalty0 (1):\penalty0 1015,
  2018{\natexlab{b}}.

\bibitem[Vorobeychik(2010)]{Vorobeychik10Probabilistic}
Yevgeniy Vorobeychik.
\newblock Probabilistic analysis of simulation-based games.
\newblock \emph{ACM Transactions on Modeling and Computer Simulation},
  20\penalty0 (3), 2010.

\bibitem[Walsh et~al.(2002)Walsh, Das, Tesauro, and Kephart]{Walsh02}
William~E. Walsh, Rajarshi Das, Gerald Tesauro, and Jeffrey~O. Kephart.
\newblock Analyzing complex strategic interactions in multi-agent games.
\newblock In \emph{AAAI Workshop on Game Theoretic and Decision Theoretic
  Agents}, 2002.

\bibitem[Walsh et~al.(2003)Walsh, Parkes, and Das]{Walsh03}
William~E. Walsh, David~C. Parkes, and Rjarshi Das.
\newblock Choosing samples to compute heuristic-strategy {Nash} equilibrium.
\newblock In \emph{Proceedings of the Fifth Workshop on Agent-Mediated
  Electronic Commerce}, 2003.

\bibitem[Wellman(2006)]{Wellman06}
Michael~P. Wellman.
\newblock Methods for empirical game-theoretic analysis.
\newblock In \emph{Proceedings of The National Conference on Artificial
  Intelligence and the Innovative Applications of Artificial Intelligence
  Conference}, 2006.

\bibitem[Wellman et~al.(2013)Wellman, Kim, and Duong]{Wellman13}
Michael~P. Wellman, Tae~Hyung Kim, and Quang Duong.
\newblock Analyzing incentives for protocol compliance in complex domains: A
  case study of introduction-based routing.
\newblock In \emph{Proceedings of the 12th Workshop on the Economics of
  Information Security}, 2013.

\bibitem[Wiedenbeck and Wellman(2012)]{WiedenbeckW12}
Bryce Wiedenbeck and Michael~P. Wellman.
\newblock Scaling simulation-based game analysis through deviation-preserving
  reduction.
\newblock In \emph{Proceedings of the International Conference on Autonomous
  Agents and Multiagent Systems (AAMAS)}, 2012.

\bibitem[Zhou et~al.(2017)Zhou, Li, and Zhu]{zhou2017identify}
Yichi Zhou, Jialian Li, and Jun Zhu.
\newblock Identify the {N}ash equilibrium in static games with random payoffs.
\newblock In \emph{Proceedings of the International Conference on Machine
  Learning (ICML)}, 2017.

\end{thebibliography}



\newpage

\renewcommand{\thepage}{}

\begin{appendices}


\setcounter{section}{0}
\renewcommand{\thesection}{\Alph{section}}

\begin{center}
    \textbf{\Large{Appendices:\\Multiagent Evaluation under Incomplete Information}}
\end{center}

We provide here supplementary material that may be of interest to the reader.
Note that section and figures in the main text that are referenced here are clearly indicated via numerical counters (e.g., \cref{fig:uncertaintyillustration}), whereas those in the appendix itself are indicated by alphabetical counters (e.g., \cref{fig:ucb-ue-symmetric}).

\section{Related Work}\label{sec:related_work}

Originally, Empirical Game Theory was introduced to reduce and study the complexity of large economic problems in electronic commerce, e.g., continuous double auctions~\cite{Walsh02,Walsh03,Wellman06}, and later it has been also applied in various other domains and settings  \cite{PhelpsPM04,PhelpsCMNPS07,PonsenTKR09,Hennes13,Tuyls18}.
Empirical game theoretic analysis and the effects of uncertainty in payoff tables (in the form of noisy payoff estimates and/or missing table elements) on the computation of Nash equilibria have been studied for some time \citep{jordan2008searching,WiedenbeckW12,Vorobeychik10Probabilistic,sakota2019learning,Aghassi2006,fearnley2015learning,zhou2017identify}, with contributions including sample complexity bounds for accurate equilibrium estimation \citep{Vorobeychik10Probabilistic}, adaptive sampling algorithms \citep{zhou2017identify}, payoff query complexity results of computing approximate Nash equilibria in various types of games \citep{fearnley2015learning}, and the formulation of particular varieties of equilibria robust to noisy payoffs \citep{Aghassi2006}. These earlier methods are mainly based on the Nash equilibrium concept and use, amongst others, information-theoretic ideas (value of information) and regression techniques to generalize payoffs of strategy profiles. By contrast, in this paper we focus on both Elo ratings and an approach inspired by response graphs, evolutionary dynamics and Markov-Conley Chains, capturing the underlying dynamics of the multiagent interactions and providing a rating of players on their long-term behavior \cite{omidshafiei2019alpha}.

The Elo rating system was originally introduced to rate chess players and named after Arpad Elo \cite{elo1978rating}. It defines a measure to express the relative strength of a player, and as such has also been widely adopted in machine learning to evaluate the strength of agents or strategies \citep{DSilverHMGSDSAPL16,silver2017mastering,Tuyls18}. Unfortunately, when applying Elo rating in machine learning, and multiagent learning particular, Elo is problematic: it is restricted to 2-player interactions, it is unable to capture intransitive behaviors and an Elo score can potentially be artificially inflated \citep{balduzzi2018re}. A Bayesian skill rating system called TrueSkill, which handles player skill uncertainties and generalized Elo rating, was introduced in \citet{herbrich:07}. For an introduction and discussion of extensions to Elo rating see, e.g., \citet{Coulom08}. Other researchers have also introduced a method based on a fuzzy pair-wise comparison matrix that uses a cosine similarity measure for ratings, but this approach is also limited to two-player interactions \citep{CHAO2018}.

Another recent work that inherently uses response graphs as its underlying dynamical model is the PSRO algorithm (Policy-Space Response Oracles) \cite{lanctot2017unified}. The Deep Cognitive Hierarchies model relates PSRO to cognitive hierarchies, and is equivalent to a response graph. The algorithm is essentially a generalization of the Double Oracle algorithm \citep{McMahanGB03} and Fictitious Self-Play \citep{Heinrich15FSP}, iteratively computing approximate best responses to the meta-strategies of other agents. 

\section{\alpharank: Additional Background}\label{sec:alpharanktheory}

This section provides additional background on the \alpharank ranking algorithm.

Given match outcomes for a $K$-player game, \alpharank computes rankings as follows:
\begin{enumerate}
    \item Construct meta-payoff tables $\mathbf{M}^k$ for each player $k \in \{1,\ldots,K\}$ (e.g., by using the win/loss ratios for the different strategy/agent match-ups as payoffs)
    \item Compute the transition matrix $\mathbf{C}$, as detailed in \cref{sec:preliminaries}
    \item Compute the stationary distribution, $\boldsymbol{\pi}$, of $\mathbf{C}$
    \item Compute the agent rankings by ordering the masses of $\boldsymbol{\pi}$ 
\end{enumerate}

In the transition structure outlined in \cref{sec:preliminaries} Eq.~\cref{eq:alpharanktransition1}, the factor $(\sum_{l=1}^K (|S^l| - 1))^{-1}$ normalizes across the different strategy profiles that $s$ may transition to, whilst the second factor represents the relative fitness of the two profiles $s$ and $\sigma$
In practice, $m \in \mathbb{N}$ is typically fixed and one considers the invariant distribution $\pi_\alpha$ as a function of the parameter $\alpha$.
\Cref{fig:fixprob} illustrates the fixation probabilities in the \alpharank model, for various values of $m$ and $\alpha$.

\paragraph{Finite-$\alpha$ limit.} In general, the invariant distribution tends to converge as $\alpha \rightarrow \infty$, and we take $\alpha$ to be sufficiently large such that $\pi_\alpha$ has effectively converged and corresponds to the MCC solution concept.

\paragraph{Infinite-$\alpha$ limit.} An alternative approach is to set $\alpha$ infinitely large, then introduce a small perturbation along every edge of the response graph, such that transitions can occur from dominated strategies to dominant ones. This perturbation enforces irreducibility of the Markov transition matrix $\mathbf{C}$, yielding a unique stationary distribution and corresponding ranking. 

\begin{figure}[h!]
    \centering
    \null\hfill
    \begin{subfigure}{.33\textwidth}
        \includegraphics[width=\textwidth]{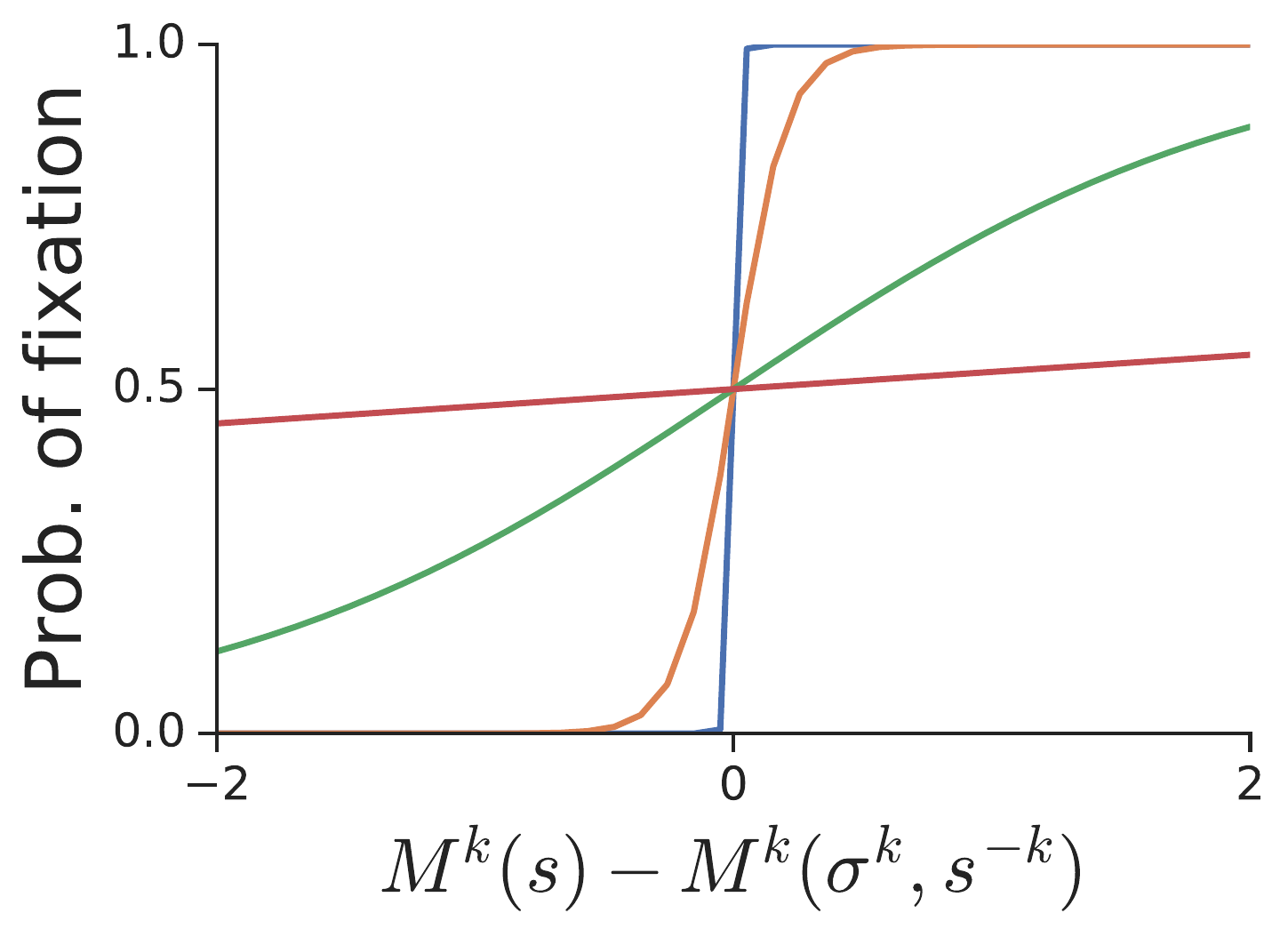}
        \caption{\alpharank population size $m=2$.}
    \end{subfigure}
    \hfill
    \begin{subfigure}{.33\textwidth}
        \includegraphics[width=\textwidth]{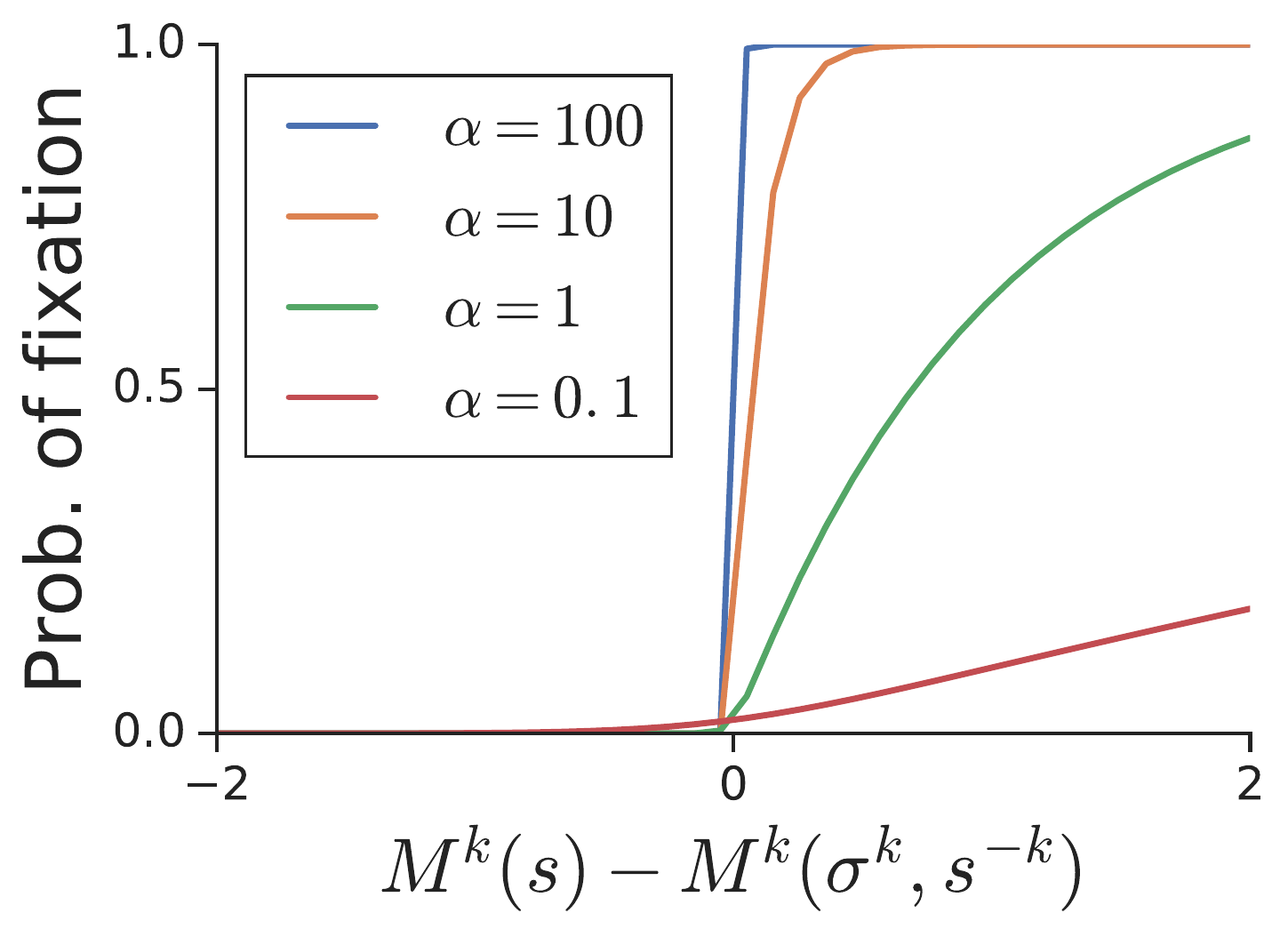}
        \caption{\alpharank population size $m=50$.}
    \end{subfigure}
    \hfill\null
    \caption{Illustrations of fixation probabilities in the \alpharank model.}
    \label{fig:fixprob}
\end{figure}

\begin{figure}[h!]
    \centering
    \begin{subtable}[c]{0.24\textwidth}
        \begin{align*}
            \begin{array}{cc|ccc}
            & & \multicolumn{2}{c}{\text{II}} \\
            & & \text{L} & \text{C} & \text{R} \\ \hline
            \multirow{3}{*}{I}  & \text{U} & 2,1 & 1,2 & 0,0 \\
             & \text{M} & 1,2 & 2,1 & 1,0\\
             & \text{D} & 0,0 & 0,1 & 2,2\\
            \end{array} \, 
        \end{align*}
        \vspace{-7pt}
        \caption{}
        \label{table:example_mcc_payoffs}
    \end{subtable}
    \begin{subfigure}[c]{.32\textwidth}
        \centering
        \includegraphics[keepaspectratio,width=0.8\textwidth,page=1]{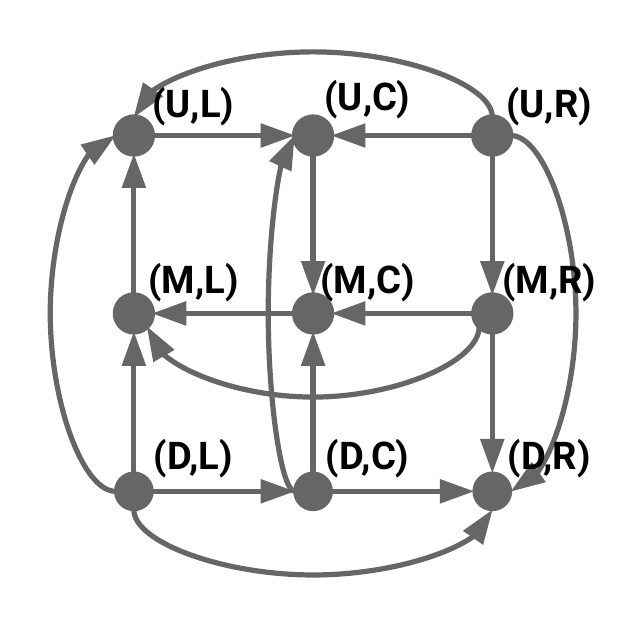}
        \caption{}
        \label{fig:example_mcc_response_graph}
    \end{subfigure}
    \hfill%
    \begin{subfigure}[c]{.32\textwidth}
        \centering
        \includegraphics[keepaspectratio,width=0.8\textwidth,page=2]{img/mcc_example.pdf}
        \vspace{-15pt}
        \caption{}
        \label{fig:example_mcc_response_graph_mccs}
    \end{subfigure}
    \caption{
    The response graph associated to payoffs shown in  \subref{table:example_mcc_payoffs} is visualized in  \subref{fig:example_mcc_response_graph}.
    \subref{fig:example_mcc_response_graph_mccs} MCCs associated with the response graph highlighted in blue. 
    }
\end{figure}

The \textbf{response graph} of a game is a directed graph where nodes correspond to pure strategy profiles, and directed edges if the deviating player’s new strategy is a better-response. The response graph for the game specified in \cref{table:example_mcc_payoffs} is illustrated in \cref{fig:example_mcc_response_graph}.  

\textbf{Markov-Conley Chains (MCCs)} are defined as the sink strongly connected components of the response graph. The MCCs associated with the payoffs specified in \cref{table:example_mcc_payoffs} are illustrated in \cref{fig:example_mcc_response_graph_mccs}.
The stationary distribution computed by \alpharank corresponds to a ranking of strategy profiles in the MCCs of the game response graph, indicating the average amount of time individuals in the underlying evolutionary model spend playing each strategy profile. 

\section{Elo Rating System: Overview and Theoretical Results}\label{sec:elo_rating}

This section provides an overview of the Elo rating system, along with theoretical guarantees on the number of samples needed to construct accurate payoff matrices using Elo.

\subsection{Elo Evaluation}
Consider games involving two players with shared strategy set $S^1$.
Elo computes a vector $\mathbf{r} \in \mathbb{R}^{S^1}$ quantifying the strategy ratings.
Let $\phi(x) = (1 + \exp(-x))^{-1}$, then the probability of $s^1 \in S^1$ beating $s^2 \in S^1$ predicted by Elo is $\mathbf{q}_{s^1,s^2}(\mathbf{r}) = \phi(\mathbf{r}_{s^1}-\mathbf{r}_{s^2})$.
Consider a batch of $N$ two-player game outcomes $(s^1_n, s^2_n, u_n)_{n=1}^N$, where $\{s^1_n,s^2_n\}\in S^1$ are the player strategies and $u_n$ is the observed (noisy) payoff to player $1$ in game $n$. 
Let $\mathbf{u} \in \mathbb{R}^{N}$ denote the vector of all observed payoffs, and denote by \emph{BatchElo} the algorithm applied to the batch of outcomes. 
BatchElo fits ability parameters $\mathbf{r}$ by minimizing the following objective with respect to $\mathbf{r}$:
\begin{align}
    L_{\mathrm{Elo}}(\mathbf{r}; \mathbf{u})=\sum_{n=1}^N - u_n \log \left(\phi(\mathbf{r}_{\mathbf{s}_n^1} - \mathbf{r}_{\mathbf{s}_n^2}) \right) - (1- u_n) \log \left(1-\phi(\mathbf{r}_{\mathbf{s}_n^1} - \mathbf{r}_{\mathbf{s}_n^2})\right) \, .
\end{align}
Ordering the elements of $\mathbf{r}$ gives the strategy rankings. 
Yet despite its widespread use for ranking \citep{lai2015giraffe,arneson2010monte,silver2017mastering,gruslys2017reactor}, Elo has no predictive power in intransitive games (e.g., Rock-Paper-Scissors) \citep{balduzzi2018re}. 

\subsection{Theoretical Results}\label{sec:elo_theoretical_results}
In analogy with the sample complexity results for \alpharank presented in \cref{sec:sample_complexity}, we give the following result on the sample complexity of Elo ranking, building on the work of \citet{balduzzi2018re}.
\begin{restatable}{theorem}{thmEloComplexity}\label{thm:elo_complexity}
    Consider a symmetric, two-player win-loss game with finite strategy set $S^1$ and payoff matrix $\mathbf{M}$. Let $\mathbf{q}$ be the fitted payoffs obtained from the BatchElo model on the payoff matrix $\mathbf{M}$, and let $\hat{\mathbf{q}}$ be the fitted payoffs obtained from the BatchElo model on an empirical payoff table $\hat{\mathbf{M}}$, based on $N_{s, s^\prime}$ interactions between each pair of strategies $s, s^\prime$. If we take, for each pair of strategy profiles $s, s^\prime \in S^1$, a number of interactions $N_{s, s^\prime}$ satisfying
    \begin{align}
        N_{s, s^\prime} > 0.5|S^1|^2 \varepsilon^{-2} \log(|S^1|^2/\delta )  \, .
    \end{align}
    Then it follows that with probability at least $1-\delta$,
    \begin{align}
        \left|\sum_{s^\prime} \left( \mathbf{q}_{s, s^\prime} - \hat{\mathbf{q}}_{s, s^\prime} \right) \right| < \varepsilon 
    \qquad \forall s \in S^1.
    \end{align}
\end{restatable}

\begin{proof}
As in \citet[Proposition 1]{balduzzi2018re}, we have that the row and column sums of $\hat{\mathbf{q}},\mathbf{q}$ match those of $\hat{\mathbf{p}},\mathbf{p}$, respectively. 
Thus, as a first result we obtain
\begin{align*}
    \sum_{s^\prime} (\mathbf{q}_{s, s^\prime} - \hat{\mathbf{q}}_{s, s^\prime} )
      &  =
    \sum_{s^\prime} (\mathbf{q}_{s,s^\prime} - \mathbf{p}_{s, s^\prime}) + 
          \sum_{s^\prime} (\mathbf{p}_{s, s^\prime} - \hat{\mathbf{p}}_{s, s^\prime}) + 
          \sum_{s^\prime} (\hat{\mathbf{p}}_{s, s^\prime} - \hat{\mathbf{q}}_{s, s^\prime}) \\
      &  =
    \sum_{s^\prime} (\mathbf{p}_{s, s^\prime} - \hat{\mathbf{p}}_{s, s^\prime} ) \quad \forall s \, , \qquad
\end{align*}
By analogous calculation, we obtain the following result for column sums:
\begin{align*}
    \sum_{s^\prime} (\mathbf{q}_{s, s^\prime} - \hat{\mathbf{q}}_{s, s^\prime} )
        =
    \sum_{s^\prime} (\mathbf{p}_{s, s^\prime} - \hat{\mathbf{p}}_{s, s^\prime} ) \quad \forall s \, .
\end{align*}
We may now apply Hoeffding's inequality to each $\hat{\mathbf{p}}_{s, s^\prime}$ with at least $N_{s,s^\prime}$ samples as in the statement of the theorem, and applying a union bound yields the required inequality.
\end{proof}

\section{Proofs of results from \cref{sec:sample_complexity}}\label{sec:supp-samplecomplexityproofs}

\subsection{Proof of \cref{thm:sample_complexity}}

\SampleComplexity*

We begin by stating and proving several preliminary results.

\begin{restatable}[Finite-$\alpha$ confidence bounds]{theorem}{ConfBoundsPi}
\label{thm:conf_bounds_pi}
    Suppose all payoffs are bounded in the interval $[-M_\mathrm{max}, M_\mathrm{max}]$.
    Let $0 < \varepsilon < \frac{g(\alpha, \eta, m, M_\mathrm{max})}{ 2^{|S|} L(\alpha, M_\mathrm{max})}$, and let $\hat{\mathbf{M}}$ be an empirical payoff table, such that
    \begin{align}
        \sup_{\substack{k \in [K] \\ s \in \prod_l S^l}} |\mathbf{M}^k(s^k, s^{-k}) - \hat{\mathbf{M}}^k(s^k, s^{-k})| \leq \varepsilon \, .
    \end{align}
    Then, denoting the invariant distribution of the Markov chain associated with $\hat{\mathbf{M}}$ by $\hat{\boldsymbol{\pi}}$, we have
    \begin{align}
        \max_{s \in \prod_k S^k} |\boldsymbol{\pi}(s) - \hat{\boldsymbol{\pi}}(s)| \leq 18 \varepsilon \frac{ L(\alpha, M_\mathrm{max})}{g(\alpha, \eta, m, M_\mathrm{max})} \sum_{n=1}^{|S|-1} \binom{|S|}{n} n^{|S|} \, .
    \end{align}
\end{restatable}

We base our proof of \cref{thm:conf_bounds_pi} on the following corollary of \citep[Theorem 1]{PerturbedMarkovChains}.
\begin{theorem}\label{thm:perturbedmarkovchains}
Let $q$ be an irreducible transition kernel on a finite state space $S$, with invariant distribution $\pi$. 
Let $\beta \in (0, 1/2^{|S|})$, and let $\hat{q}$ be another transition kernel on $S$. Suppose that $|q(t|s) - \hat{q}(t|s)| \leq \beta q(t|s)$ for all states $s,t \in S$. If $q$ is irreducible, then $\hat{q}$ is irreducible, and the invariant distributions $\pi, \hat{\pi}$ of $q, \hat{q}$ satisfy
\begin{align*}
    |\pi(s) - \hat{\pi}(s)| \leq 18 \pi(s) \beta \sum_{n=1}^{|S|-1} \binom{|S|}{n} n^{|S|} \, ,
\end{align*}
for all $s \in S$.
\end{theorem}

We next derive several technical bounds on properties of the \alpharank transition matrix $\mathbf{C}$, defined in \cref{eq:alpharanktransition1}.

\begin{lemma}\label{lem:lowerbound}
    Suppose all payoffs are bounded in the interval $[-M_\mathrm{max}, M_\mathrm{max}]$. 
    Then all non-zero elements of the transition matrix $\mathbf{C}$ are lower-bounded by $g(\alpha, \eta, m, M_\mathrm{max})$.
\end{lemma}

\begin{proof}
    Consider first an off-diagonal, non-zero element of the matrix. The transition probability is given by
    \begin{align*}
        \eta \left( \frac{1 - \exp(-\alpha x)}{1 - \exp(-\alpha m x)} \right) \, ,
    \end{align*}
    for some $x \in [-2M_\mathrm{max}, 2M_\mathrm{max}]$, by assumption of boundedness of payoffs. This quantity is minimal for $x=-2M_\mathrm{max}$, and we hence obtain the required lower-bound. We note also that these transition probabilities are upper-bounded by taking $x=2M_\mathrm{max}$, yielding an upper bound of $\eta \frac{1 - \exp(-2\alpha M_\mathrm{max})}{1 - \exp(-2\alpha m M_\mathrm{max})}$.
    The transition probability of a diagonal element $\mathbf{C}_{ii}$ takes the form $1 - \sum_{j \not= i} \mathbf{C}_{ij}$. There are $\eta^{-1}$ non-zero terms in the sum, each of which is upper-bounded by $\eta \frac{1 - \exp(-2\alpha M_\mathrm{max})}{1 - \exp(-2\alpha m M_\mathrm{max})}$. Hence, we obtain the following lower bound on the diagonal entries:
    \begin{align*}
        1 - \eta^{-1} \eta \frac{1 - \exp(-2\alpha M_\mathrm{max})}{1 - \exp(-2\alpha m M_\mathrm{max})} 
        = & 1 - \frac{1 - \exp(-2\alpha M_\mathrm{max})}{1 - \exp(-2\alpha m M_\mathrm{max})} \\
        = & \frac{ 1 - \exp(-2\alpha m M_\mathrm{max}) - 1 + \exp(-2\alpha M_\mathrm{max})}{1 - \exp(-2\alpha m M_\mathrm{max})} \\
        = & \frac{ \exp(2\alpha (m-1) M_\mathrm{max}) - 1}{\exp(2\alpha m M_\mathrm{max}) - 1} \\
        \geq & \eta \frac{ \exp(2\alpha M_\mathrm{max}) - 1}{\exp(2\alpha m M_\mathrm{max}) - 1} \, ,
    \end{align*}
    as required.
\end{proof}

\begin{lemma}\label{lem:lipschitzbound}
    Suppose all payoffs are bounded in the interval $[-M_\mathrm{max}, M_\mathrm{max}]$. 
    All transition probabilities are Lipschitz continuous as a function of the collection of payoffs $(\hat{\mathbf{M}}^k(s^k, s^{-k}) | k \in [K], s \in \prod_l S^l)$ under the infinity norm, with Lipschitz constant upper-bounded by $L(\alpha, M_\mathrm{max})$.
\end{lemma}
\begin{proof}
    We begin by considering off-diagonal, non-zero elements. The transition probability takes the form $\eta \left( \frac{1 - \exp(-\alpha x)}{1 - \exp(-\alpha m x)} \right)$, for some $x \in [-2M_\mathrm{max}, 2M_\mathrm{max}]$, representing the payoff difference for the pair of strategy profiles concerned. First, the Lipschitz constant of $x \mapsto \exp(-\alpha x)$ on the domain $x \in [-2M_\mathrm{max}, 2M_\mathrm{max}]$ is $\alpha\exp(2\alpha M_\mathrm{max})$. Composing this function with the function $x \mapsto \eta\frac{1 -x}{1 - x^m}$ yields the transition probability, and this latter function has Lipschitz constant $\eta$ on $(0,\infty)$. Therefore, the function $x \mapsto \eta \left( \frac{1 - \exp(-\alpha x)}{1 - \exp(-\alpha m x)} \right)$ is Lipschitz continuous on $[-2M_\mathrm{max}, 2M_{\mathrm{max}}]$, with Lipschitz constant upper-bounded by $\eta \alpha \exp(2\alpha M_\mathrm{max})$.
    Hence, the Lipschitz constant of the off-diagonal transition probabilities as a function of the payoffs under the infinity norm is upper-bounded by $2 \eta \exp(2\alpha M_\mathrm{max})$.
    Turning our attention to the diagonal elements, we may immediately read off their Lipschitz constant as being upper-bounded by $\eta^{-1} \times 2 \eta \alpha \exp(2\alpha M_\mathrm{max}) = 2 \alpha \exp(2\alpha M_\mathrm{max})$, and hence the statement of the lemma follows.
\end{proof}

We can now give the full proof of \cref{thm:conf_bounds_pi}.
\begin{proof}[Proof of \cref{thm:conf_bounds_pi}]
By Lemma~\ref{lem:lipschitzbound}, we have that all elements of the transition matrix $C$ are Lipschitz with constant $L(\alpha, M_{\max})$ with respect to the payoffs $(\mathbf{M}^k(s^k, s^{-k}) | k \in [K], s \in \prod_l S^l)$ under the infinity norm. Thus, denoting the transition matrix constructed from the empirical payoff table $\hat{\mathbf{M}}$ by $\hat{\mathbf{C}}$, we have the following bound for all $i,j$:
\begin{align*}
    |\mathbf{C}_{ij} - \hat{\mathbf{C}}_{ij}| \leq \varepsilon L(\alpha, M_\mathrm{max}) \, .
\end{align*}
Next, we have that all non-zero elements of $\mathbf{C}_{ij}$ are lower-bounded by $g(\alpha, \eta, m, M_\mathrm{max})$ by Lemma~\ref{lem:lowerbound}, and hence we have
\begin{align*}
    |\mathbf{C}_{ij} - \hat{\mathbf{C}}_{ij}| \leq \varepsilon L(\alpha, M_\mathrm{max}) \leq \varepsilon \frac{L(\alpha, M_\mathrm{max})}{g(\alpha, \eta, m, M_\mathrm{max})}\mathbf{C}_{ij} \, .
\end{align*}
By assumption, the coefficient of $\mathbf{C}_{ij}$ on the right-hand side is less than $1/2^{|S|}$. We may now appeal to Theorem \ref{thm:perturbedmarkovchains}, to obtain
\begin{align*}
    |\pi(s) - \hat{\pi}(s)| \leq 18 \pi(s) \varepsilon \frac{L(\alpha, M_\mathrm{max})}{g(\alpha, \eta, m, M_\mathrm{max})} \sum_{n=1}^{|S|-1} \binom{|S|}{n} n^{|S|} \, ,
\end{align*}
for all $s \in \prod_k S^k$. Using the trivial bound $\pi(s) \leq 1$ for each $s \in \prod_k S^k$ yields the result.
\end{proof}

With \cref{thm:conf_bounds_pi} established, we can now prove \cref{thm:sample_complexity}.
\begin{proof}[Proof of \cref{thm:sample_complexity}]
    By Theorem~\ref{thm:conf_bounds_pi}, we have that $\max_{s \in S} |\pi(s) - \hat{\pi}(s) | < \varepsilon$ is guaranteed if
    \begin{align*}
        \max_{\substack{k \in [K]\\ s \in S}} |\mathbf{M}^k(s^k, s^{-k}) - \hat{\mathbf{M}}^k(s^k, s^{-k})| < \frac{\varepsilon g(\alpha, \eta, m, M_\mathrm{max}) }{ 18 L(\alpha, M_\mathrm{max}) \sum_{n=1}^{|S|-1} \binom{|S|}{n}n^{|S|} } < \frac{g(\alpha, \eta, m, M_\mathrm{max})}{2^{|S|} L(\alpha, M_\mathrm{max})} \, .
    \end{align*}
    We separate this into two conditions. Firstly, from the second inequality above, we require
    \begin{align*}
        \varepsilon < \frac{18 \sum_{n=1}^{|S|-1} \binom{|S|}{n} n^{|S|}}{2^{|S|}} \, ,
    \end{align*}
    which is satisfied by assumption. 
    Secondly, we have the condition
    \begin{align*}
        \max_{\substack{k \in [K]\\ s \in S}} |\mathbf{M}^k(s^k, s^{-k}) - \hat{\mathbf{M}}^k(s^k, s^{-k})| < \frac{\varepsilon g(\alpha, \eta, m, M_\mathrm{max}) }{ 18 L(\alpha, M_\mathrm{max}) \sum_{n=1}^{|S|-1} \binom{|S|}{n}n^{|S|} } \, .
    \end{align*}
    Now, write $N_{s}$ for the number of trials with the strategy profile $s \in S$. We will next use the following form of Hoeffding's inequality:
    Let $X_1,\ldots,X_N$ be i.i.d. random variables supported on $[a, b]$. Let $\varepsilon > 0$ and $\delta > 0$. Then for $N > (b-a)^2\log(2/\delta)/ (2\varepsilon^2)$, we have $\mathbb{P}\left(\left|\frac{1}{N}\sum_{n=1}^N X_n - \mathbb{E}\left\lbrack X_1\right\rbrack \right| > \varepsilon\right) < \delta$.
    Applying this form of Hoeffding's inequality to the random variable $\hat{M}^k(s^k, s^{-k})$, if we take 
    \begin{align*}
        N_{s} > \frac{4M^2_\mathrm{max}\log(2K|S|/\delta)}{2 \left(\frac{\varepsilon g(\alpha, \eta, m, M_\mathrm{max}) }{L(\alpha, M_\mathrm{max}) 18 \sum_{n=1}^{|S|-1} \binom{|S|}{n}n^{|S|} }\right)^2}
            =
        \frac{648 M_{\mathrm{max}}^2 \log(2K|S|/\delta) L(\alpha, M_\mathrm{max})^2 \left( \sum_{n=1}^{|S|-1} \binom{|S|}{n} n^{|S|} \right)^2}{\varepsilon^2 g(\alpha, \eta, m, M_\mathrm{max})^2    } \, ,
    \end{align*}
    then 
    \begin{align*}
        |\mathbf{M}^k(s^k, s^{-k}) -  \hat{\mathbf{M}}^k(s^k, s^{-k})| < \frac{\varepsilon g(\alpha, \eta, m, M_\mathrm{max}) }{ 18 L(\alpha, M_\mathrm{max}) \sum_{n=1}^{|S|-1} \binom{|S|}{n}n^{|S|} }
    \end{align*}
    holds with probability at least $1-\delta/(|S|K)$. Applying a union bound over all $k \in [K]$ and $s \in S$ then gives
    \begin{align*}
        \max_{\substack{k \in [K]\\ s \in S}} |\mathbf{M}^k(s^k, s^{-k}) - \hat{\mathbf{M}}^k(s^k, s^{-k})| < \frac{\varepsilon g(\alpha, \eta, m, M_\mathrm{max}) }{18 L(\alpha, M_\mathrm{max}) \sum_{n=1}^{|S|-1} \binom{|S|}{n}n^{|S|} }
    \end{align*}
    with probability at least $1-\delta$, as required.
\end{proof}

\subsection{Proof of \cref{thm:sample_complexityinfinite}}

\ExactRecoverySampleComplexity*

We begin by stating and proving a preliminary result.

\begin{restatable}[Infinite-$\alpha$ confidence bounds]{theorem}{ExactRecovery}\label{thm:exactrecovery}
    Suppose all payoffs are bounded in $[-M_\mathrm{max}, M_\mathrm{max}]$. Suppose that for all $k \in [K]$ and for all $s^{-k} \in S^{-k}$, we have $|\mathbf{M}^k(\sigma, s^{-k}) - \mathbf{M}^k(\tau, s^{-k})| \geq \Delta$ for all distinct $\sigma, \tau \in S^k$, for some $\Delta > 0$. Then if $|\hat{\mathbf{M}}^k(s) - \mathbf{M}^k(s)| < \Delta/2$ for all $s \in \prod_l S^l$ and all $k \in [K]$, then we have that $\hat{\mathbf{C}} = \mathbf{C}$.
\end{restatable}
\begin{proof}
    From the inequality $|\hat{\mathbf{M}}^k(s) - \mathbf{M}^k(s)| < \Delta/2$ for all $s \in S$, we have by the triangle inequality that $|(\hat{\mathbf{M}}^k(\sigma, s^{-k}) - \hat{\mathbf{M}}(\tau, s^{-k})) - (\mathbf{M}^k(\sigma, s^{-k}) - \mathbf{M}^k(\tau, s^{-k})) | < \Delta$ for all $k \in [K]$, $s^{-k} \in S^{-k}$, and all distinct $\sigma, \tau \in S^k$. Thus, by the assumption of the theorem, $\hat{\mathbf{M}}^k(\sigma, s^{-k}) - \hat{\mathbf{M}}(\tau, s^{-k})$ has the same sign as $\mathbf{M}^k(\sigma, s^{-k}) -\mathbf{M}^k(\tau, s^{-k})$ for all $k \in [K]$, $s^{-k} \in S^{-k}$, and all distinct $\sigma, \tau \in S^k$.
    It therefore follows from the expression for fixation probabilities (assuming the Fermi revision protocol), in the limit as $\alpha \rightarrow \infty$, the estimated transition probabilites $\hat{\mathbf{C}}$ exactly match the true transition probabilities $\mathbf{C}$, and hence the invariant distribution computed from the empirical payoff table matches that computed from the true payoff table.
\end{proof}

With this result in hand, we may now prove \cref{thm:sample_complexityinfinite}.

\begin{proof}[Proof of \cref{thm:sample_complexityinfinite}]
    We use the following form of Hoeffding's inequality. Let $X_1,\ldots,X_N$ be i.i.d. random variables supported on $[a,b]$, and let $\varepsilon > 0$, $\delta > 0$. Then if $N > (b-a)^2 \log(2/\delta)/(2 \varepsilon^2)$, we have $\mathbb{P}(|\frac{1}{N}\sum_{n=1}^N X_i - \mathbb{E}[X_1]| > \varepsilon) < \delta$. Applying this form of Hoeffding's inequality to an empirical payoff $\hat{\mathbf{M}}^k(s)$, and writing $|S|=\Pi_k |S^k|$, we obtain the result that for
    \begin{align*}
        N_s > \frac{(2M_\mathrm{max})^2 \log(2|S|K/\delta)}{2(\Delta/2)^2} = \frac{8M_\mathrm{max}^2 \log(2|S|K/\delta)}{\Delta^2} \, ,
    \end{align*}
    we have
    \begin{align*}
        |\hat{\mathbf{M}}^k(s) - \mathbf{M}^k(s)| < \Delta/2 \, ,
    \end{align*}
    with probability at least $1 - \delta/(|S|K)$. Applying a union bound over all $k \in [K]$ and all $s \in \prod_l S^l$, we obtain that if
    \begin{align*}
        N_s >  \frac{8M_\mathrm{max}^2 \log(2|S|K/\delta)}{\Delta^2} \qquad \forall s \in \prod_k S^k \, ,
    \end{align*}
    then by Theorem \ref{thm:exactrecovery}, we have that the transition matrix $\hat{\mathbf{C}}$ computed from the empirical payoff table $\hat{\mathbf{M}}$ matches the transition matrix $\mathbf{C}$ corresponding to the true payoff table $\mathbf{M}$ with probability at least $1-\delta$.
\end{proof}

\section{Proofs of results from \cref{sec:adaptive_sampling_algs}}\label{sec:supp-adaptivesampling}

\thmAdaptiveCorrectness*
\begin{proof}
We begin by introducing some notation. For a general strategy profile $s\in S$, denote the empirical estimator of $\mathbf{M}^k(s)$ after $u$ interactions by $\hat{\mathbf{M}}^k_u(s)$ and let $n_t(s)$ be the number of interactions of~$s$ by time $t$, and finally let $L(\hat{\mathbf{M}}^k_{u}(s), \delta, u, t)$ (respectively $U(\hat{\mathbf{M}}^k_{u}(s), \delta, u, t)$) denote the lower (respectively upper) confidence bound for $\mathbf{M}^k(s)$ at some time index $t$ after $u$ interactions of~$s$, empirical estimator $\hat{\mathbf{M}}^k_u(s)$, and confidence parameter $\delta$. We remark that in typical pure exploration problems, $t$ counts the total number of interactions; in our scenario, since we have a \textit{collection} of best-arm identification problems, we take a separate time index $t$ for each problem, counting the number of interactions for strategy profiles concerned with each specific problem. Thus, for the best-arm identification problem concerning two strategy profiles $s$, $s^\prime$, with interaction counts $n_s$, $n_{s^\prime}$, we take $t = n_s + n_{s^\prime}$.

We first apply a union bound over each best-arm identification problem:
\begin{align*}
    \mathbb{P}(\text{Incorrect output}) \leq \!\!\!\!\! \sum_{(\sigma, s^{-k}), (\tau, s^{-k})}\!\!\!\!\! \mathbb{P}(\text{Incorrect comparison for strategy profiles } (\sigma, s^{-k})\text{ and } (\tau, s^{-k})) \, .
\end{align*}
A standard analysis can now be applied to each best-arm identification problem, following the approach of e.g., \citet{kalyanakrishnan2012pac,gabillon2012best,karnin2013almost}. 
To reduce notational clutter, we let $s \triangleq (\sigma, s^{-k})$ and $s^\prime \triangleq (\tau, s^{-k})$. Further, without loss of generality taking $\mathbf{M}^k(s) > \mathbf{M}^k(s^\prime)$, we have
\begin{align*}
     \mathbb{P}&(\text{Incorrect ordering of } s, s^\prime) \\
    &\leq  \mathbb{P}(\exists t,\, u \leq t \text{ s.t. } \mathbf{M}^k(s) < L(\hat{\mathbf{M}}^k_u(s), \delta, u, t) \text{ or } \mathbf{M}^k(s^\prime) > U(\hat{\mathbf{M}}^k_u(s^\prime), \delta, u, t) ) \\
    &\leq  \sum_{t=1}^\infty \sum_{u=1}^t \bigg\lbrack \mathbb{P}(\mathbf{M}^k(s) < L(\hat{\mathbf{M}}^k_u(s), \delta, u, t)) + \mathbb{P}(\mathbf{M}^k(s^\prime) > U(\hat{\mathbf{M}}^k_u(s^\prime), \delta, u, t)) \bigg\rbrack \, .
\end{align*}
Note that the above holds for \textit{any} sampling strategy $\mathcal{S}$. 
We may now apply an individual concentration inequality to each of the terms appearing in the sum above, to obtain
\begin{align*}
    \mathbb{P}(\text{Incorrect ordering of } s, s^\prime) \leq 2 \sum_{t=1}^\infty \sum_{u=1}^t f(u, \delta, (|S^k|)_k, t) \, ,
\end{align*}
where $f(u, \delta, (|S^k|)_k, t)$ is an upper bound on the probability of a true mean lying outside a confidence interval based on $u$ interactions at time $t$. 
Thus, overall we have
\begin{align*}
    \mathbb{P}(\text{Incorrect output}) \leq \frac{|S| \sum_{k=1}^K (|S^k| - 1)}{2} \sum_{t=1}^\infty \sum_{u=1}^t 2f(u, \delta, (|S^k|)_k, t) \, .
\end{align*}
If $f$ is chosen such that $\frac{|S| \sum_{k=1}^K (|S^k| - 1)}{2} \sum_{t=1}^\infty \sum_{u=1}^t 2f(u, \delta, (|S^k|)_k, t) \leq \delta$, then the proof of correctness is complete. It is thus sufficient to choose
\begin{align*}
    f(u,\delta, (|S^k|)_k, t) = \frac{6\delta}{\pi^2|S|\sum_{k=1}^K (|S^k| - 1) t^3} \, .
\end{align*}
Note that this analysis has followed without prescribing the particular \emph{form} of confidence interval used, as long as its coverage matches the required bounds above.
\end{proof}

\thmAdaptiveSampleComplexity*

\begin{proof}
We adapt the approach of \citet{even2006action}, and use the notation introduced in the proof of \cref{thm:adaptivecorrectness} First, let $\bar{U}(\delta, u, t) = \sup_x \left\lbrack U(x, \delta, u, t) - x\right\rbrack$, and $\bar{L}(\delta, u, t) = \sup_x \left\lbrack x - L(x, \delta, u, t)\right\rbrack$. Note that if we have counts $n_s=u$ and $n_{s^\prime}=v$ such that
\begin{align}\label{eq:samplecomplexitycondition}
     \mathbf{M}^k(s) - \mathbf{M}^k(s^\prime) > 2\bar{U}(\delta, u, t) + 2 \bar{L}(\delta, v, t) \, ,
\end{align}
then we have
\begin{align*}
    L(\hat{\mathbf{M}}^k_u(s), \delta, u, t) - U(\hat{\mathbf{M}}^k_v(s^\prime), \delta, v, t) & > \hat{\mathbf{M}}^k_u(s) - \bar{L}(\delta, u, t) - (\hat{\mathbf{M}}^k_v(s^\prime) + \bar{U}(\delta, v, t)) \\
    & \overset{(a)}{>} \mathbf{M}^k(s) - 2\bar{L}(\delta, u, t) - \mathbf{M}^k(s^\prime) - 2\bar{U}(\delta, v, t) \\
    & > 0 \, ,
\end{align*}
where (a) holds with probability at least $1 - 2f(u, \delta, (|S^k|)_k, t)$. Hence, with probability at least $1-2f(u, \delta, (|S^k|)_k, t)$ the algorithm must have terminated by this point. Thus, if $u$, $v$ and $t$ are such that \cref{eq:samplecomplexitycondition} holds, then we have that the algorithm will have terminated with high probability. Writing $\Delta = \mathbf{M}^k(s) - \mathbf{M}^k(s^\prime)$, with all observed outcomes bounded in $[-M_\text{max},M_\text{max}]$, we have
    \begin{align*}
        \bar{U}(\delta, u, t) = \bar{L}(\delta, u, t) = \sqrt{\frac{4M_\text{max}^2\log(2/f(u, \delta, (|S^k|)_k, t))}{u}} \, .
    \end{align*}
    We thus require
    \begin{align*}
        \Delta > 2\sqrt{\frac{4M_\text{max}^2\log(2/f(u, \delta, (|S^k|)_k, t))}{u}} + 2\sqrt{\frac{4M_\text{max}^2\log(2/f(v, \delta, (|S^k|)_k, t))}{v}} \, .
    \end{align*}
    Taking $u=v$, and using $f(u, \delta, (|S^k|)_k, t) = \frac{6\delta}{\pi^2|S|\sum_{k=1}^K (|S^k| - 1) t^3}$ as above, we obtain the condition
    \begin{align*}
        \Delta > 4\sqrt{\frac{4M_\text{max}^2}{u} \log\left(\frac{8 u^3 \pi^2 |S| \sum_{k=1}^K (|S^k|-1)}{3\delta}  \right)} \, .
    \end{align*}
    A sufficient condition for this to hold is $u = \mathcal{O}(\Delta^{-2} \log(\frac{2}{\delta \Delta} ))$. Thus, if all strategy profiles $s$ have been sampled at least $\mathcal{O}(\Delta^{-2} \log(\frac{2}{\delta \Delta} ))$ times, the algorithm will have terminated with probability at least $1 - 2\delta$. Up to a $\log(1/\Delta)$ factor, this matches the instance-aware bounds obtained in the previous section.
\end{proof}

\section{Additional material on \responsegraphucb}\label{sec:supp-responsegraphucb}

In this section, we give precise details of the form of the confidence intervals considered in the ResponseGraphUCB algorithm, described in the main paper.

\paragraph{Hoeffding bounds (UCB).} In cases where the noise distribution on strategy payoffs is known to be bounded on an interval $[a,b]$, we can use confidence bounds based on the standard Hoeffding inequality. For a confidence level $\delta$ and count index $n$, and mean estimate $\overline{x}$, this interval takes the form $(\overline{x} - \sqrt{(b-a)^2\log(2/\delta) / 2n}, \overline{x} + \sqrt{(b-a)^2\log(2/\delta)/2n})$. Optionally, an additional exploration bonus based on a time index $t$, measuring the total number of samples for all strategy profiles concerned in the comparison, can be added, yielding an interval of the form $(\overline{x} - \sqrt{(b-a)^2\log(2/\delta)f(t) / n}, \overline{x} + \sqrt{(b-a)^2\log(2/\delta)f(t)/n})$, for some function $f : \mathbb{N} \rightarrow (0, \infty)$.

\paragraph{Clopper-Pearson bounds (CP-UCB).} In cases where the noise distribution is known to be Bernoulli, it is possible to tighten the Hoeffding confidence interval described above, which is valid for any distribution supported on a fixed finite interval. The result is the asymmetric Clopper-Pearson confidence interval: for an empirical estimate $\overline{x}$ formed from $n$ samples, at a confidence level $\delta$, the Clopper-Pearson interval \citep{garivier2011kl,clopper1934use} takes the form $(B(\delta/2; n\overline{x}, n -n\overline{x} + 1), B(1-\delta/2; n\overline{x}+1, n -n\overline{x} )$, where $B(p;v,w)$ is the $p$\textsuperscript{th} quantile of a Beta($v,w$) distribution.

\paragraph{Relaxed variants.} As an alternative to waiting for confidence intervals to become fully disjoint before declaring an edge comparison to be resolved, we may instead stipulate that confidence intervals need only $\varepsilon$-disjoint (that is, the length of their intersection is $< \varepsilon$). This has the effect of reducing the number of samples required by the algorithm, and may be practically advantageous in instances where the noise distributions do not attain the worst case under the confidence bound (for example, low-variance noise under the Hoeffding bounds); clearly however, such an adjustment breaks any theoretical guarantees of high-probability correct comparisons.

\section{Additional material on uncertainty propagation}\label{sec:uncertaintysupp}

In this section, we provide details for the high-level approach outlined in Section~\ref{sec:uncertainty}, in particular giving more details regarding the reduction to response graph selection (in particular, selecting directions of particular edges within the response graph), and then using the PageRank-style reduction to obtain a CSSP policy optimization problem.

\textbf{Reduction to edge direction selection.} The infinite-$\alpha$ \alpharank output is a function of the payoff table $\mathbf{M}$ only through the infinite-$\alpha$ limit of the corresponding transition matrix $\mathbf{C}$ defined in \eqref{eq:alpharanktransition1};  
this limit is determined by binary payoff comparisons for pairs of strategy profiles differing in a single strategy. We can therefore summarize the set of possible transition matrices $\mathbf{C}$ which are compatible with the payoff bounds $\mathbf{L}$ and $\mathbf{U}$ by compiling a list $E$ of response graph edges for which payoff comparisons (i.e., response graph edge directions) are uncertain under $\mathbf{L}$ and $\mathbf{U}$. Note that it may be possible to obtain even tighter confidence intervals on $\pi(s)$ by keeping track of which combinations of directed edges in $E$ are compatible with an underlying payoff table $\mathbf{M}$ itself, but by not doing so we only broaden the space of possible response graphs, and hence still obtain valid confidence bounds. 
The confidence interval for $\pi(s)$ could thus be obtained by computing the output of infinite-$\alpha$ \alpharank for each transition matrix $\mathbf{C}$ that arises from all choices of edge directions for the uncertain edges in $E$. 
However, this set is generally exponentially large in the number of strategy profiles, and thus intractable to compute. The next step is to reduce this problem to one which is solvable using standard dynamic programming techniques to avoid this intractability.

\textbf{Reduction to CSSP policy optimization.} We now use a reduction similar to that used in the PageRank literature for optimizing stationary distribution mass \citep{csaji2014pagerank}, encoding the problem above as an SSP optimization problem. 
For a transition matrix $\mathbf{C}$, let $(X_t)_{t=0}^\infty$ denote the corresponding Markov chain over the space of strategy profiles $S$, and define the \emph{mean return times} $\boldsymbol{\lambda} \in [0, \infty]^{S}$ by $\lambda(u) = \mathbb{E}\left\lbrack \inf\{ t > 0 | X_t = u\} | X_0 = u \right\rbrack$, for each $u \in S$.
By basic Markov chain theory, when $\mathbf{C}$ is such that $s$ is recurrent, 
the mass attributed to $s$ under the stationary distribution supported on the MCC containing~$s$ is equal to $1/\lambda(s)$; 
thus, maximizing (respectively, minimizing) $\pi(s)$ over a set of transition matrices is equivalent to minimizing (respectively, maximizing) $\lambda(s)$. 
Define the mean hitting time of $s$ starting at $u$ for all $u \in S$ by $\boldsymbol{\varphi} \in [0, \infty]^S$, where $\varphi(u) = \mathbb{E}\left\lbrack \inf\{ t > 0 | X_t = s\} | X_0 = u \right\rbrack$,  whereby $\varphi(s) = \lambda(s)$; 
then $\boldsymbol{\varphi} = \widetilde{\mathbf{C}} \boldsymbol{\varphi} + \mathbf{1}$, where $\widetilde{\mathbf{C}}$ is the substochastic matrix given by setting the column of $\mathbf{C}$ corresponding to state $s$ to the zero vector, and $\mathbf{1} \in \mathbb{R}^S$ is the vector of ones. 

Note that $\boldsymbol\varphi$ has the interpretation of a value function in an SSP problem, wherein the absorbing state is $s$ and all transitions before absorption incur a cost of $1$. 
The original problem of maximizing (respectively, minimizing) $\pi(s)$ is now expressed as minimizing (respectively, maximizing) this value at state $s$ over the set of compatible transition matrices $\widetilde{\mathbf{C}}$. 
We turn this into a standard control problem by specifying the \emph{action set} at each state $u \in S$ as $\mathcal{P}(\{ e \in E | u \in e \})$, the powerset of the set of uncertain edges in $E$ incident to $u$; the interpretation of selecting a subset $U$ of these edges is that precisely the edges in $U$ will be selected to flow \emph{out} of $u$; 
this then fully specifies the row of $\widetilde{\mathbf{C}}$ corresponding to $u$. 
Crucially, the action choices cannot be made independently at each state; 
if at state $u$, the uncertain edge between $u$ and $u^\prime$ is chosen to flow in a particular direction, then at state $u^\prime$ a \emph{consistent} action must be chosen, so that the actions at both states agree on the direction of the edge, thus leading to a \emph{constrained} SSP optimization problem. We refer to this problem as $\texttt{CSSP}(S, \mathbf{L}, \mathbf{U}, s)$.
While general solution of CSSPs is intractable, we recall the statement of Theorem~\ref{thm:cssp} that it is sufficient to consider the \emph{unconstrained} version of $\texttt{CSSP}(S, \mathbf{L}, \mathbf{U}, s)$ to recover the same optimal policy.

\thmCSSP*

We conclude by restating the final statements of Section~\ref{sec:uncertainty}.
In summary, the general approach for finding worst-case upper and lower bounds on infinite-$\alpha$ \alpharank ranking weights $\pi(s)$ for a given strategy profile $s \in S$ is to formulate the unconstrained SSP described above, find the optimal policy (using, e.g., linear programming, policy or value iteration), and then use the inverse relationship between mean return times and stationary distribution probabilities in recurrent Markov chains to obtain the bound on the ranking weight $\pi(s)$ as required. In the single-population case, the SSP problem can be run with the infinite-$\alpha$ \alpharank transition matrices; in the multi-population case, a sweep over perturbation levels in the infinite-$\alpha$ \alpharank model can be performed, as with standard \alpharank, to deal with the possibility of several sink strongly-connected components in the response graph.

\subsection{MCC detection}

Here, we outline a straightforward algorithm for determining whether $\inf_{\mathbf{L} \leq \hat{\mathbf{M}} \leq \mathbf{U}} \pi_{\hat{\mathbf{M}}}(s) = 0$, without recourse to the full CSSP reduction described in Section \ref{sec:uncertainty}. First, we use $\mathbf{L}$ and $\mathbf{U}$ to split the edges of the response graph into two disjoint sets $E_U$, edges for which the direction is uncertain under $\mathbf{L}$ and $\mathbf{U}$, and $E_C$, the edges with certain direction. We then construct the set $F_A \subseteq S$ of \emph{forced ancestors} of $s$; that is, the set of strategy profiles that can reach $s$ using a path of edges contained in $E_C$, including $s$ itself. We also define the set $F_D\subseteq S$ of \emph{forced descendents} of $s$; that is, the set of strategy profiles that can be reached from $s$ using a path of edges in $E_C$, including $s$ itself. If $F_D \not\subseteq F_A$, then $s$ can clearly be made to lie outside an MCC by setting all edges in $E_U$ incident to $F_D \setminus F_A$ to be directed \emph{into} $F_D$. Then there are no edges directed out of $F_D \setminus F_A$, so this set contains at least one MCC. There also exists a path from $s$ to $F_D \setminus F_A$, and hence $s$ cannot lie in an MCC, so $\inf_{\mathbf{L} \leq \hat{\mathbf{M}} \leq \mathbf{U}} \pi_{\hat{\mathbf{M}}}(s) = 0$. If, on the other hand, $F_D \subseteq F_A$, we may set all uncertain edges between $F_A$ and its complement to be directed away from $F_A$. We then iteratively compute two sets: $F_{A,\text{out}}$, the set of profiles in $F_A$ for which there exists a path out of $F_A$, and its complement $F_A \setminus F_{A,\text{out}}$. Any uncertain edges between these two sets are then set to be directed towards $F_{A,\text{out}}$, and the sets are then recomputed. This procedure terminates when either $F_{A,\text{out}} = F_A$, or there are no uncertain edges left between $F_A$ and $F_{A,\text{out}}$. If at this point there is no path from $s$ out of $F_A$, we conclude that $s$ must lie in an MCC, and so $\inf_{\mathbf{L} \leq \hat{\mathbf{M}} \leq \mathbf{U}} \pi_{\hat{\mathbf{M}}}(s) > 0$, whilst if such a path does exist, then $s$ does not lie in an MCC, so $\inf_{\mathbf{L} \leq \hat{\mathbf{M}} \leq \mathbf{U}} \pi_{\hat{\mathbf{M}}}(s) = 0$.

\subsection{Proof of \cref{thm:cssp}}\label{sec:uncertainty_prop_proof}

\thmCSSP*

\begin{proof}
Let $\ssmone$ be the substochastic matrix associated with the optimal \emph{unconstrained} policy, and suppose there are two action choices that are inconsistent; that is, there exist strategy profiles $u$ and $v$ differing only in index $k$, such that either (i) at state $u$, the edge direction is chosen to be $u \to v$, and at state $v$, the edge direction is chosen to be $v \to u$; or (ii) at state $u$, the edge direction is chosen to be $v \to u$, and at state $v$, the edge direction is chosen to be $u \to v$. We show that in either case, there is a policy without this inconsistency that achieves at least as good a value of the objective as the inconsistent policy.

We consider first case (i). Let $\boldsymbol{\varphi}$ be the associated expected costs under the inconsistent optimal policy, and suppose without loss of generality that $\boldsymbol{\varphi}(v) \geq \boldsymbol{\varphi}(u)$. Let $\ssmtwo$ be the substochastic matrix obtained by adjusting the action at state $u$ so that the edge direction between $u$ and $v$ is $v \to u$, consistent with the action choice at $v$. Denote the expected costs under this new transition matrix $\widetilde{\mathbf{D}}$ by $\boldsymbol{\mu}$. We can compare $\boldsymbol{\varphi}$ and $\boldsymbol{\mu}$ via the following calculation. By definition, we have $\boldsymbol{\varphi} = \ssmone\boldsymbol{\varphi} + \mathbf{1}$ and $\boldsymbol{\mu} = \ssmtwo \boldsymbol{\mu} + \mathbf{1}$. Thus, we compute
\begin{align*}
    \boldsymbol{\varphi} - \boldsymbol{\mu} & = (\ssmone\boldsymbol{\varphi} + \mathbf{1} ) - (\ssmtwo \boldsymbol{\mu} + \mathbf{1} ) \\
     & = \ssmone\boldsymbol{\varphi} - \ssmtwo \boldsymbol{\mu} \\
     & = \ssmone\boldsymbol{\varphi} - \ssmtwo\boldsymbol{\varphi} + \ssmtwo\boldsymbol{\varphi} - \ssmtwo \boldsymbol{\mu} \\
     & = (\ssmone - \ssmtwo)\boldsymbol{\varphi} + \ssmone(\boldsymbol{\varphi} - \boldsymbol{\mu}) \\
     \implies \boldsymbol{\varphi} - \boldsymbol{\mu} & = (\mathbf{I} - \ssmtwo)^{-1} (\ssmone - \ssmtwo) \boldsymbol{\varphi} \, .
\end{align*}
In this final line, we assume that $\mathbf{I}-\ssmtwo$ is invertible. If it is not, then it follows that $\ssmtwo$ is a strictly stochastic matrix, thus corresponding to a policy in which no edges flow into $s$.
From this we immediately deduce that the minimal value of $\boldsymbol{\varphi}(s)$ is $\infty$; hence, we may assume $\mathbf{I}-\ssmtwo$ is invertible in what follows. 
Assume for now that $s \not\in\{u,v\}$. Now note that $\ssmone$ and $\ssmtwo$ differ only in two elements: $(u,u)$, and $(u,v)$, and thus the vector $(\ssmone - \ssmtwo) \boldsymbol{\varphi}$ has a particularly simple form; all coordinates are $0$, except coordinate $u$, which is equal to $\eta(\boldsymbol{\varphi}(v) - \boldsymbol{\varphi}(u)) \geq 0$. Finally, observe that all entries of $(\mathbf{I} - \ssmtwo)^{-1} = \sum_{k=0}^\infty \ssmtwo^k$ are non-negative, and hence we obtain the element-wise inequality $\boldsymbol{\varphi} - \boldsymbol{\mu} \geq 0$, proving that the policy associated with $\ssmtwo$ is at least as good as $\ssmone$, as required. The argument is entirely analogous in case (ii), and when one of the strategies concerned is $s$ itself. Thus, the proof is complete.
\end{proof}

\section{Additional empirical details and results}

\subsection{Experimental procedures and reproducibility}\label{sec:experiment_procedures}

We detail the experimental procedures here.

The results shown in \cref{fig:uncertainty_example_soccer} are generated by computing the upper and lower payoff bounds given a mean payoff matrix and confidence interval size for each entry, then running the procedure outlined in~\cref{sec:uncertainty}.

As \cref{fig:ucb-ue} shows an intuition-building example of the \responsegraphucb outputs, it was computed by first constructing the payoff table specified in the figure, then running \responsegraphucb with the parameters specified in the caption.
The algorithm was then run until termination, with the strategy-wise sample counts in \cref{fig:ucb-ue} computed using running averages. 

The finite-$\alpha$ \alpharank results in \cref{fig:N_s_empirical} for every combination of $\alpha$ and $\epsilon$ are computed using 20, 5, and 5 independent trials, respectively, for the Bernoulli, soccer, and poker meta-games.
The same number of trials applies for every combination of $\delta$ and \responsegraphucb in \cref{fig:ranking_error_comparisons}.

The ranking results shown in \cref{fig:ranking_error_comparisons} are computed for 10 independent trials for each game and each $\delta$.

The parameters swept in our plots are the error tolerance, $\delta$, and desired error $\epsilon$. The range of values used for sweeps is indicated in the respective plots in~\cref{sec:experiments}, with end points chosen such that sweeps capture both the high-accuracy/high-sample complexity and low-accuracy/low-sample complexity regimes.

The sample mean is used as the central tendency estimator in plots, with variation indicated as the 95\% confidence interval that is the default setting used in the \texttt{Seaborn} visualization library that generates our plots. No data was excluded and no other preprocessing was conducted to generate these plots.

No special computing infrastructure is necessary for running \responsegraphucb, nor for reproducing our plots; we used local workstations for our experiments.

\clearpage 
\subsection{Full comparison plots}\label{sec:full_comparison_plots}

As noted in \cref{sec:adaptive_sampling_algs}, sample complexity and ranking error under adaptive sampling are of particular interest. 
To evaluate this, we consider all variants of \responsegraphucb in \cref{fig:sweep_sampling_and_confidence_full}.

\begin{figure}[h!]
    \newcommand{\figHeight}{8.5\baselineskip}
    \subcaptionbox{Bernoulli games.
        }[0.25\textwidth]{ \includegraphics[height=\figHeight]{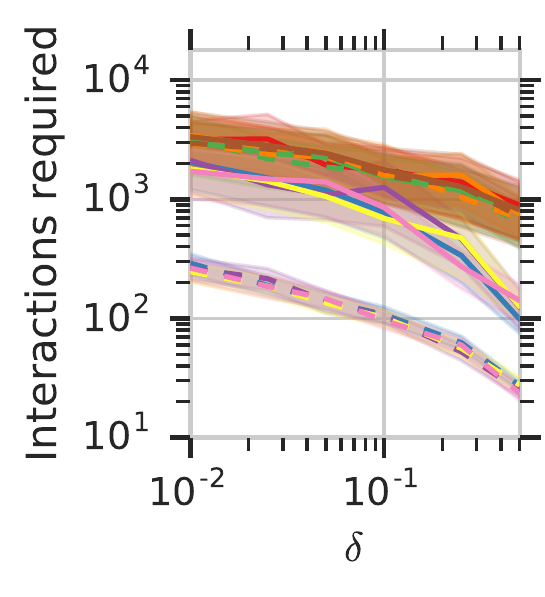}\\[-6pt]
        \includegraphics[height=\figHeight]{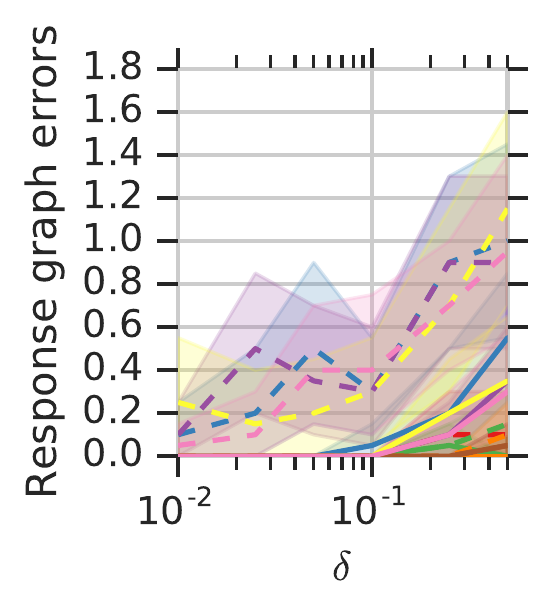}
        \vspace{-8pt}
    }%
    \rulesep%
    \subcaptionbox{Soccer meta-game.
    }[0.25\textwidth]{
        \includegraphics[height=\figHeight]{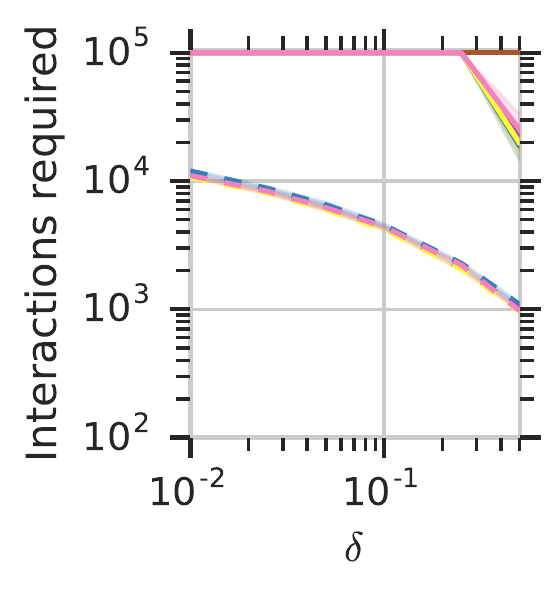}\\[-6pt]
        \includegraphics[height=\figHeight]{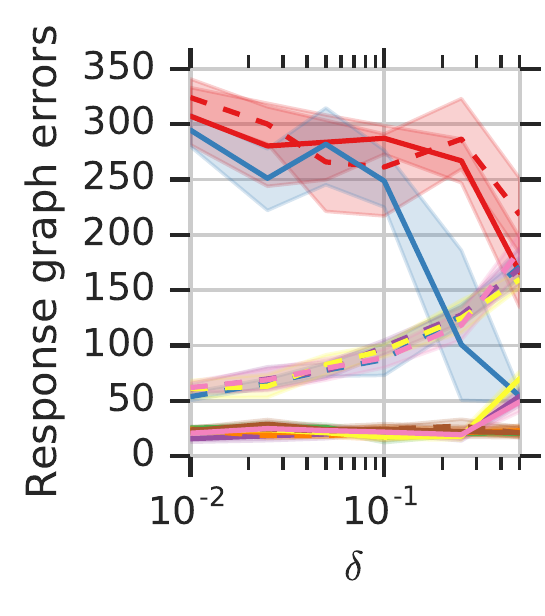}
        \vspace{-8pt}
    }%
    \rulesep%
    \subcaptionbox{Poker meta-game.
    }[0.25\textwidth]{
        \includegraphics[height=\figHeight]{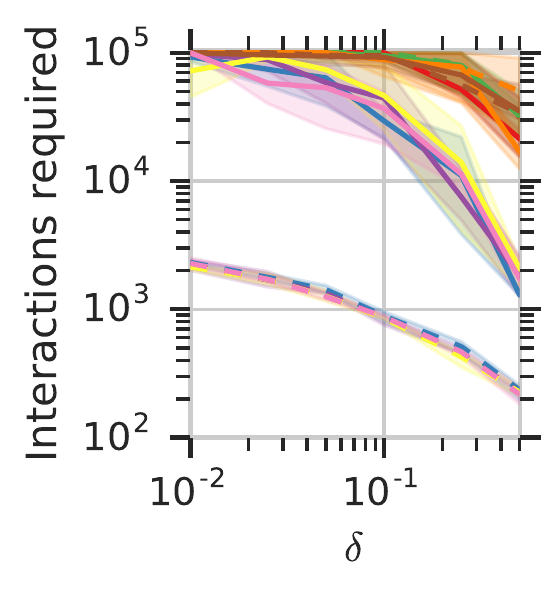}\\[-6pt]
        \includegraphics[height=\figHeight]{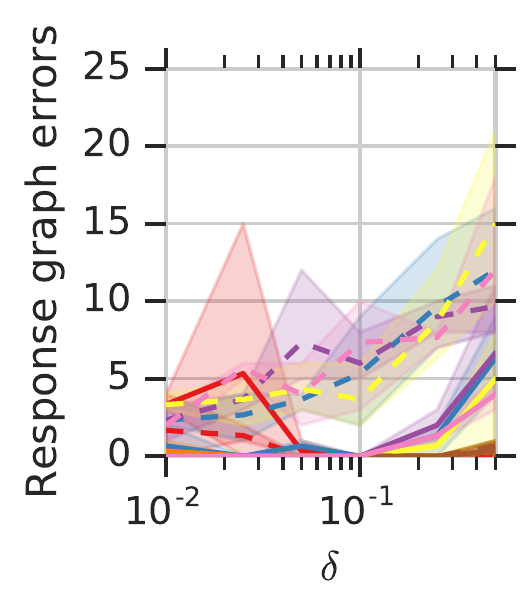}
        \vspace{-8pt}
    }
    \subcaptionbox*{}[0.21\textwidth]{
        \vspace{-5pt}
        \includegraphics[height=16\baselineskip]{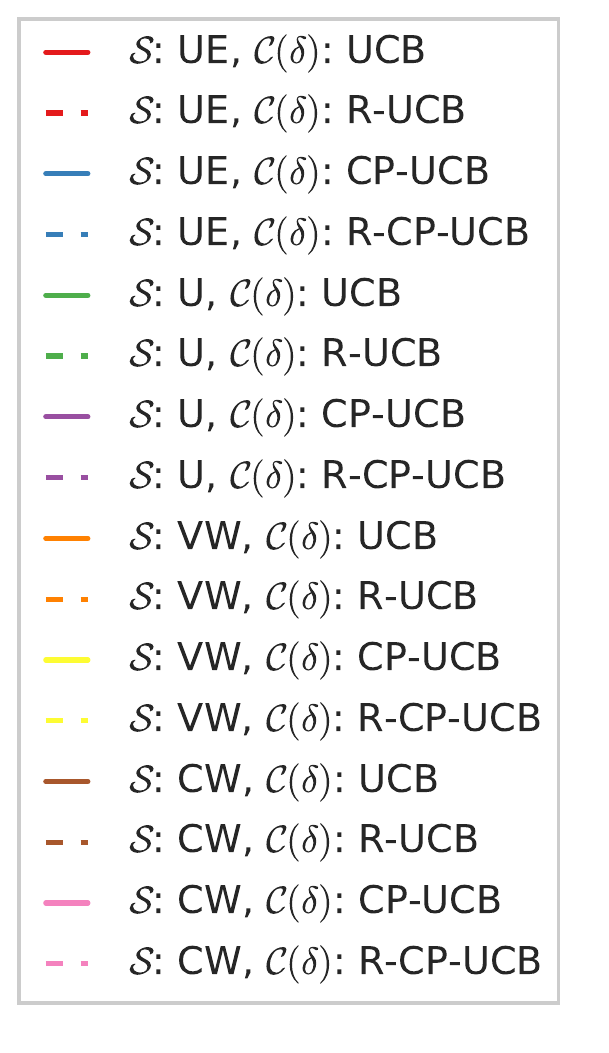}
    }
    \caption{\responsegraphucb performance metrics versus error tolerance $\delta$ for all games. First and second rows, respectively, show the \# of interactions required and response graph edge errors.}
    \label{fig:sweep_sampling_and_confidence_full}
\end{figure}

\clearpage 
\subsection{Exploiting knowledge of symmetry in games}\label{sec:symmetric_responsegraphucb}
\begin{figure}[h]
    \begin{subfigure}{.3\textwidth}
        \centering
        \includegraphics[keepaspectratio,width=\textwidth]{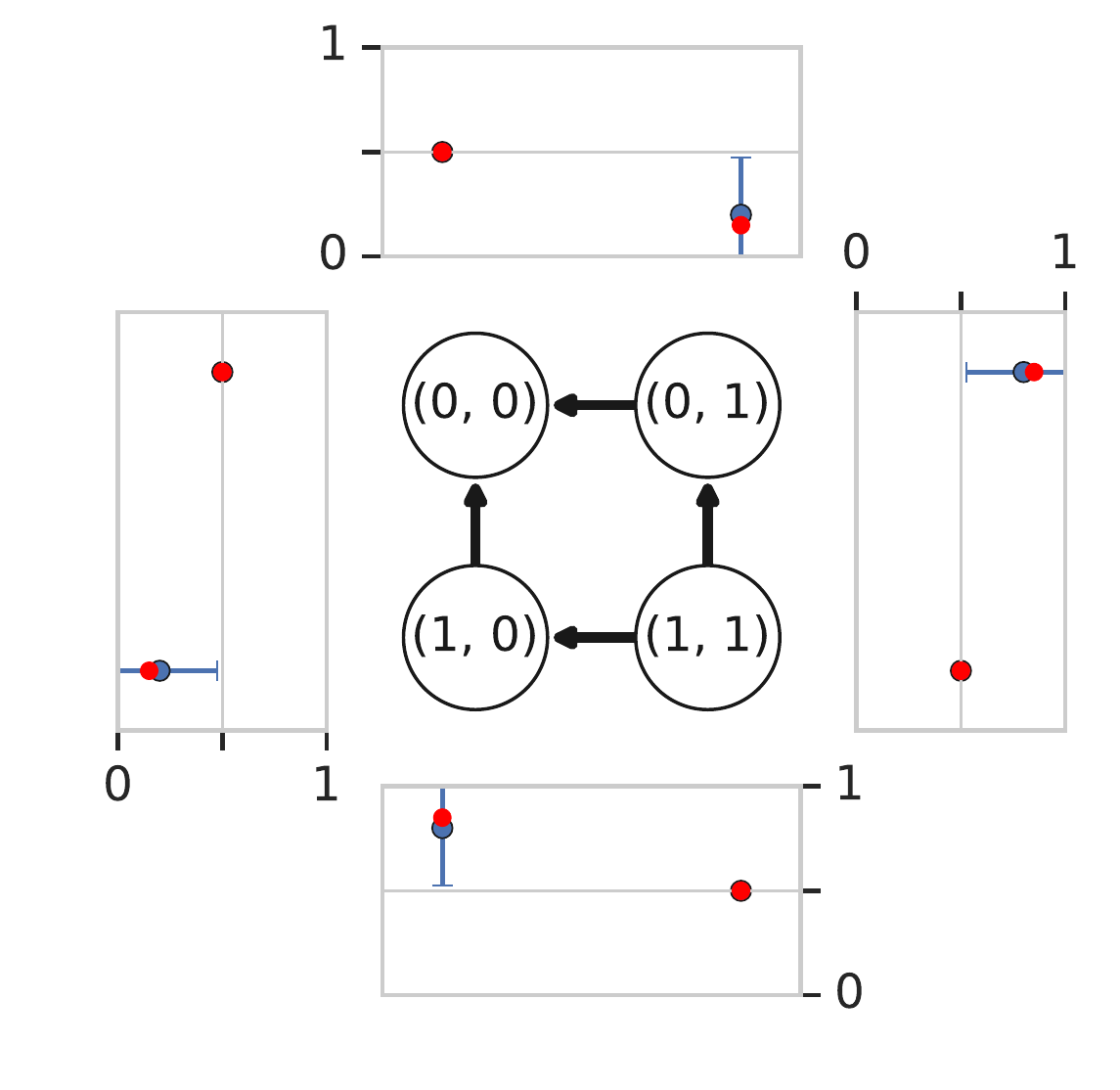}
        \caption{Reconstructed response graph.}
        \label{fig:ucb-ue-result-symmetric}
    \end{subfigure}
    \hfill%
    \begin{subfigure}{.6\textwidth}
        \centering
        \includegraphics[keepaspectratio,width=0.8\textwidth]{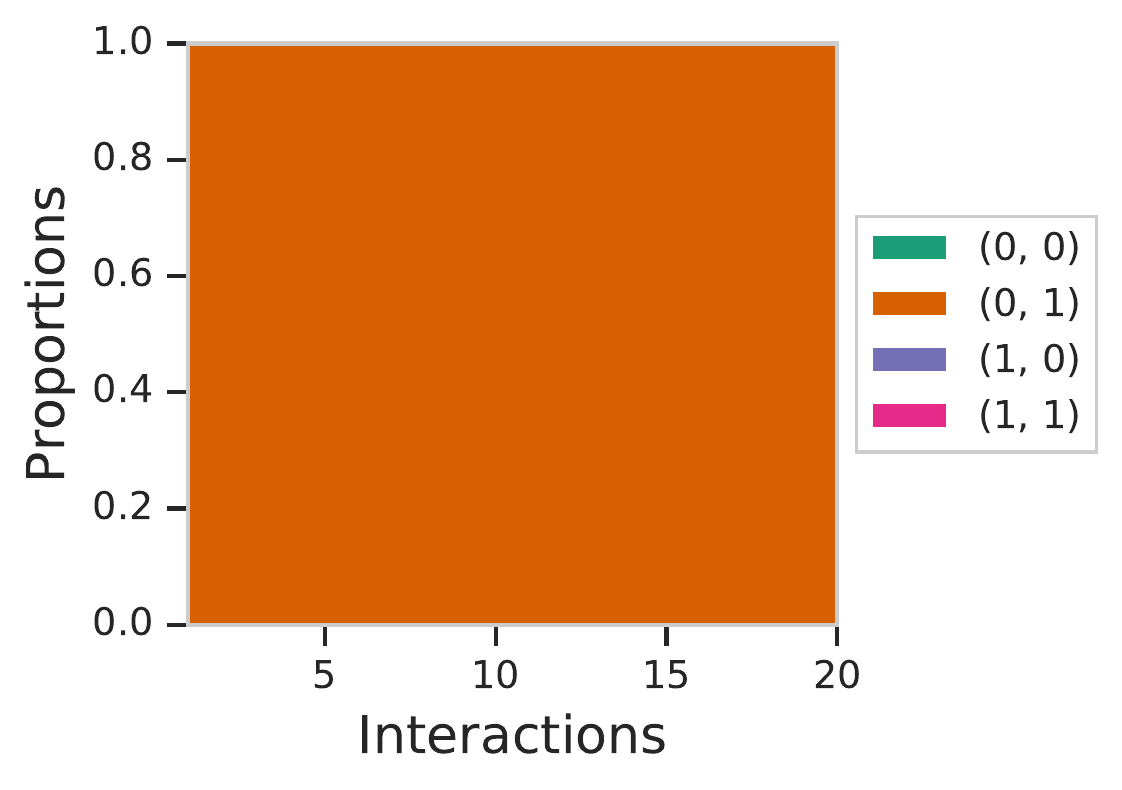}
        \caption{Strategy-wise sample counts.}
        \label{fig:ucb-ue-result-queries-symmetric}
    \end{subfigure}
    \caption{\responsegraphucb($\delta=0.1$, $\mathcal{S}=$UE, $\mathcal{C}=$UCB) evaluated on the game with payoff tables shown in  \cref{table:ucb-ue-payoff-table}, with knowledge of game symmetry exploited to reduce the total number of samples needed from 244 to 20 and sampling conducted for only a single strategy profile, $(0,1)$.}
    \label{fig:ucb-ue-symmetric}
\end{figure}

\paragraph{Symmetric games.} Let $\text{Sym}_K$ denote the symmetric group of degree $K$ over all players. 
A game is said to be symmetric if for any permutation $\rho \in \text{Sym}_K$, strategy profile $(s^1, \ldots, s^K) \in S$ and index $k \in [K]$, we have $\mathbf{M}^k(s^1,\ldots,s^K) = \mathbf{M}^{\rho(k)}(s^{\rho(1)}, \ldots, s^{\rho(K)})$. 

\paragraph{Exploiting symmetry in \responsegraphucb.} Knowledge of the symmetric constant-sum nature of a game (e.g., \cref{table:ucb-ue-payoff-table}) can significantly reduce sample complexity in \responsegraphucb:
this knowledge implies that payoffs for all symmetric strategy profiles are known a priori (e.g., payoffs for $(0,0)$ and $(1,1)$ are 0.5 in this example);
moreover, each observed outcome for a strategy profile $(s^1, \ldots, s^K)$ yields a `free' observation of $\mathbf{M}^{\rho(k)}(s^{\rho(1)}, \ldots, s^{\rho(K)})$ for all permutations $\rho \in \text{Sym}_K$, strategy profiles $(s^1, \ldots, s^K) \in S$, and players $k \in [K]$.
For the example in \cref{table:ucb-ue-payoff-table}, the symmetry-exploiting variant of the algorithm is able to reconstruct the true underlying response graph using only 20 samples of a single strategy profile $(0,1)$.

In \cref{fig:ucb-ue-symmetric}, we evaluate \responsegraphucb on the game shown in \cref{table:ucb-ue-payoff-table}, this time exploiting  the knowledge of game symmetry as discussed in \cref{sec:ucb-ue}.
Note that \cref{fig:ucb-ue-result-symmetric,fig:ucb-ue-result-queries-symmetric} should be compared, respectively, with \cref{fig:ucb-ue-result,fig:ucb-ue-result-queries} in the main paper. 
Confidence bounds corresponding to the symmetry-exploiting sampler (\cref{fig:ucb-ue-result-symmetric}) are guaranteed to be tighter than the non-exploiting sampler (\cref{fig:ucb-ue-result}), and so typically we can expect the former to require fewer interactions to arrive at a ranking conclusion with the same confidence as the latter (under the condition that the payoffs really are symmetric, as is the case in win/loss two-player games). 
This is observed in this particular example, where \cref{fig:ucb-ue-result} took 244 interactions to solve, while \cref{fig:ucb-ue-result-symmetric} took only 20 samples of a single strategy profile $(0,1)$ to correctly reconstruct the response graph.

\subsection{Kendall's distance for partial rankings}\label{sec:ranking_metrics}

We use Kendall's distance for partial rankings \citep{fagin2006comparing} when comparing two rankings, $r$ and $\hat{r}$ (e.g., as done in \cref{fig:ranking_error_comparisons})

Consider a pair of partial strategy rankings $r$ and $\hat{r}$ (i.e., wherein tied rankings are allowed).
Define a fixed parameter $p$.
The Kendall distance with penalty parameter $p$ is defined,
\begin{align*}
    K(r,\hat{r};p) = \sum_{\{i,j\} \in [|S|]} \bar{K}_{i,j}(r,\hat{r};p), 
\end{align*}
where $\bar{K}_{i,j}(r,\hat{r};p)$ is:
\begin{itemize}[leftmargin=.2in]
    \item $0$ when {$i,j$ are in distinct buckets in both $r,\hat{r}$, but in the same order (e.g., $r_i>r_j$ and $\hat{r}_i>\hat{r}_j$)}
    \item $1$ when {$i,j$ are in distinct buckets in both $r,\hat{r}$, but in the reverse order (e.g., $r_i>r_j$ and $\hat{r}_i<\hat{r}_j$)}
    \item $0$ when {$i,j$ are in the same bucket in both $r$ and $\hat{r}$}
    \item $p$ when {$i,j$ are in the same bucket in one of $r$ or $\hat{r}$, but different buckets in the other.}
\end{itemize}
It can be shown that Kendall's distance is a metric when $p\in[0.5,1]$. 
We use $p=0.5$ in our experiments.



\subsection{Preliminary experiments on collaborative filtering-based approaches}\label{sec:collab_filtering}

The pairing of bandit algorithms and \alpharank seems a natural means of computing rankings in settings where, e.g., one has a limited budget for adaptively sampling match outcomes.
Our use of bandit algorithms also leads to analysis which is flexible enough to be able to deal with $K$-player general-sum games.
However, approaches such as collaborative filtering may also fare well in their own right.
We conduct a preliminary analysis of this in here, specifically for the case of two-player win-loss games.

For such games, the meta-payoff table is given by a matrix $\mathbf{M}$ with all entries lying in $(0,1)$ (encoding loss as payoff $0$ and win as payoff $1$). Taking a matrix completion approach, we might attempt to reconstruct a low-rank approximation of the payoff table from an incomplete list of (possible noisy) payoffs, and then run \alpharank on the reconstructed payoffs. Possible candidates for the low-rank structure include: (i) the payoff matrix itself; (ii) the \emph{logit matrix} $\mathbf{L}_{ij} = \log(\mathbf{M}_{ij}/(1 - \mathbf{M}_{ij}))$; and (iii) the \emph{odds matrix} $\mathbf{O}_{ij} = \exp(\mathbf{L}_{ij})$. In particular, \citet{balduzzi2018re} make an argument for the (approximate) low-rank structure of the logit matrix in many applications of interest.

\begin{figure}[h!]
    \centering
    \includegraphics[keepaspectratio,width=\textwidth]{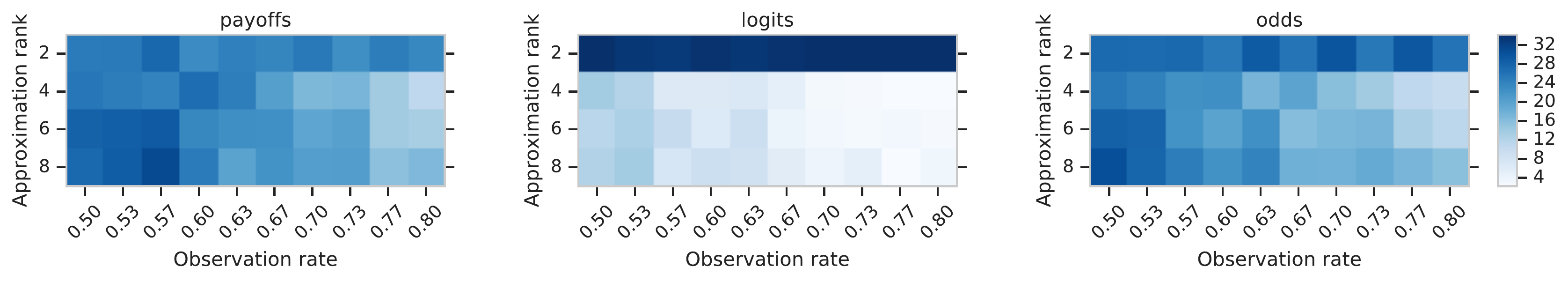}
    \caption{Ranking errors (Kendall's distance w.r.t. ground truth) from completion of, respectively, the sparse payoffs, logits, and odds matrices for Soccer dataset. 20 trials per combo of assumed matrix ranks and observation rates/density.}
    \label{fig:2dAlphaRank_ranking_err_vs_rank_and_obsrate_kendall_partial_soccer}
\end{figure}

We conduct preliminary experiments on this in \cref{fig:2dAlphaRank_ranking_err_vs_rank_and_obsrate_kendall_partial_soccer}, implementing matrix completion calculations via Alternating Minimization \citep{jain2013low}. 
We compare here the resulting \alpharank errors for the three reconstruction approaches for the Soccer meta-game.
We sweep across the observation rates of payoff matrix entries and the matrix rank assumed in the reconstruction. 
Interestingly,  conducting low-rank approximation on the logits (as opposed to the odds) matrix generally yields the lowest ranking error.
Overall, the bandit-based approach may be more suitable when one can afford to play all strategy profiles at least once, whereas matrix completion is perhaps more so when this is not feasible.
These results, we believe, warrant additional study of the performance of related alternative approaches in future work.

\end{appendices}
\end{document}